\documentclass[10pt,conference]{IEEEtran}

\usepackage{mathtools}
\usepackage[utf8]{inputenc}
\usepackage{amssymb}
\usepackage{amsthm}
\usepackage{amsmath}
\usepackage{graphicx}
\usepackage{dsfont}
\usepackage{xspace}
\usepackage[linesnumbered,ruled,noend]{algorithm2e}
\usepackage{algpseudocode}
\usepackage{url}
\usepackage{multirow}
\usepackage{color}
\usepackage{caption,subcaption}
\usepackage{wrapfig}
\usepackage{balance}

\usepackage{xcolor,colortbl}

\title{FastHare: Fast Hamiltonian Reduction for Large-scale Quantum Annealing}
\author{
\IEEEauthorblockN{Phuc Thai}
\IEEEauthorblockA{\textit{Virginia Commonwealth University} \\
thaipd@vcu.edu}
\and
\IEEEauthorblockN{My T. Thai}
\IEEEauthorblockA{\textit{University of Florida} \\
mythai@cise.ufl.edu}
\and
\IEEEauthorblockN{Tam Vu}
\IEEEauthorblockA{\textit{Oxford University} \\
tam.vu@cs.ox.ac.uk}
\and
\IEEEauthorblockN{Thang N. Dinh\thanks{Corresponding author:tndinh@vcu.edu}}
\IEEEauthorblockA{\textit{Virginia Commonwealth University} \\
tndinh@vcu.edu}
}

\IEEEoverridecommandlockouts

\newtheorem{theorem}{Theorem}[section]
\newtheorem{definition}[theorem]{Definition}

\newtheorem{lemma}[theorem]{Lemma}

\providecommand{\internalFlag}{1}
\ifnum\internalFlag=1
\newcommand{\fdef}[1]{{\color{blue}[FA: #1]}}
\else
\newcommand{\fdef}[1]{}
\fi

\begin{document}

\maketitle
% !TEX root = main.tex
\newcommand{\pnote}[1]{{\color{blue}{[Phuc: #1]}}}
\definecolor{gray}{rgb}{0.5, 0.5, 0.5}

\newcommand{\sep}{\ominus}
\newcommand{\alg}{{\textsf{ECONSET}}\xspace}
\newcommand{\flip}{{\mathsf{flip}}\xspace}
\newcommand{\flipset}{{\mathsf{find\text{-}flipping\text{-}group}\xspace}}
\newcommand{\flipsetn}{{\mathsf{search\text{-}neighbor\text{-}NGs}}\xspace}
\newcommand{\compress}{{\mathsf{compression}}\xspace}
\newcommand{\compressg}{{\mathsf{compress}}\xspace}
\newcommand{\compressm}{{\mathsf{strong\text{-}compression}}\xspace}
\newcommand{\findset}{{\mathsf{search\text{-}NGs}}\xspace}
\newcommand{\complex}{{\mathsf{dim}}\xspace}
\newcommand{\colorG}{{\mathsf{color}}\xspace}
\newcommand{\mapf}{{\mathsf{g}}\xspace}
\newcommand{\flipf}{{\mathsf{f}}\xspace}
\newcommand{\OPM}{{\mathsf{m}}\xspace}
\newcommand{\cutf}{{\omega}\xspace}
\newcommand{\cotime}[1]{\cellcolor{red!#1}#1}
\newcommand{\coreduction}[1]{\cellcolor{blue!#1}#1}
\newcommand{\wei}{{\mathsf{w}}\xspace}
\newcommand{\vwei}{{\mathbf{w}}\xspace}
\newcommand{\vsum}{{\mathsf{sum}}\xspace}
\newcommand{\cuta}{{\mathsf{c}_{|.|}}\xspace}
\newcommand{\cut}{{\mathsf{c}}\xspace}
\newcommand{\adj}{{\mathsf{adj}}\xspace}
\let\oldnl\nl
\newcommand{\nonl}{\renewcommand{\nl}{\let\nl\oldnl}}
\newcommand{\fasts}{{\mathsf{\hat{\nu}_f}}\xspace}
\newcommand{\sims}{{\mathsf{\hat{\nu}_s}}\xspace}
\newcommand{\tris}{{\mathsf{\hat{\nu}_t}}\xspace}

\newcommand{\mt}[1]{{\color{blue} {#1} -MT}}
% !TEX root = main.tex

\begin{abstract} 
Quantum annealing (QA) that encodes optimization problems into Hamiltonians remains the only near-term quantum computing paradigm that provides sufficient many qubits for real-world applications. To fit larger optimization instances on existing quantum annealers, reducing Hamiltonians into smaller equivalent Hamiltonians provides a promising approach.  Unfortunately, existing reduction techniques are either computationally expensive or ineffective in practice.   
To this end, we introduce a novel notion of \emph{non\nobreakdash-separable~group}, defined as a subset of qubits in a Hamiltonian that obtains the same value in optimal solutions. We develop non-separability theory accordingly and propose FastHare,
a highly efficient reduction method. FastHare, iteratively, detects and merges non-separable groups into single qubits. It does so within a provable worst-case time complexity of only $O(\alpha n^2)$, for some user-defined parameter $\alpha$. Our extensive  benchmarks for the feasibility of the reduction are done on both synthetic Hamiltonians and  3000+ instances from the MQLIB library.  The results show FastHare outperforms the roof duality, the implemented reduction in D-Wave's library. %, and can reduce 43\% instances \mt{size? Don't understand this. Reduce 43 instances. what is that?}. 
It demonstrates a high level of effectiveness with an average of 62\% qubits saving and 0.3s processing time, advocating for Hamiltonian reduction as \emph{an inexpensive necessity} for QA.
\end{abstract}
\section{Introduction}
The last few years has witnessed an exponential growth in quantum and quantum-inspired computing (QC) with a record number of breakthroughs \cite{arute2019quantum,honjo2021100,bharti2022noisy,mills2021two,wang2022light}.  Instead of encoding information with binary bits as in classical computing, quantum computers  use qubits to encode superposition of states \cite{bharti2022noisy} to explore exponentially combinations of states at once. QC has paved the way for much faster, more efficient solving of large-scale real-world optimization problems that  are challenging for classical computers \cite{arute2019quantum,bharti2022noisy}. 

%There is a huge potential in harvesting the power of QC for monitoring, controlling, and securing  power grid systems. Despite a few early work on for power grids \cite{ajagekar2019quantum,feng2021quantum, zhou2022noise,nikmehr2022quantum}, existing approaches are based on gate-based QCs that suffers from both the small number of qubits and the significant challenges in maintaining the decoherence among entangled qubits \cite{bharti2022noisy}. 

One promising near-term avenue for QCs  is quantum annealing (QA) \cite{kadowaki1998quantum,zhou2022noise}, a framework that incorporates algorithms and hardware designed to solve computational problems. QA leverages  quantum tunneling mechanics to perform   quantum evolution toward the ground states of final Hamiltonians that encode classical optimization problems, without necessarily insisting on universality or adiabaticity \cite{kadowaki1998quantum}.
QA is the only computing paradigm that provides a large enough number of qubits for real-world applications from RNA folding \cite{fox2021mrna,mulligan2020designing,fox2022rna}, portfolio optimization \cite{mugel2021hybrid,grozea2021optimising},  car manufacturing scheduling \cite{yarkoni2021multi} and many others \cite{mott2017solving, neukart2017traffic,kim2019leveraging}. 
In addition, the number of Qubits tends to double every 20 months over the last decade \cite{dwaveAdvantage}.
%QA remains a largely unexplored opportunity for  optimization problems in power grid systems.  In particular, it offers a huge promise in \emph{addressing the challenges in accurately and timely state estimation to provide high quality solutions that is infeasible for traditional computing to achieve within dozens of miliseconds as required control actions}.

Yet, the limited hardware resource, including the relatively small numbers of both qubits and their couplings, as well as the challenges in mapping the problem Hamiltonian on quantum processing unit (QPU) hardware  topology, aka minor-embedding \cite{choi2008minor}, pose significant challenges in scaling the QA to the real-world instances. 
For example, performing MIMO channel decoding with a 60Tx60R setup on a 64-QAM, a configuration several folds lower than the state-of-the-art hardware, will require about 11,000 physical qubits \cite{tabi2021evaluation}. This hardware requirement far exceeds the 5000+ qubits offered by the largest commercially available quantum annealer, the D-Wave Advantages platform. Thus, qubits saving techniques to reduce the hardware resource is much needed to reduce hardware resource requirement, as well as increasing the size of solvable instances on existing QPUs.

%Qubits reduction approaches and equivalence are few and not currently effective for QA \mt{ambiguous}. 

Only a few qubits reduction techniques have been studied, yet, are not effective for QA.  The most popular method is the roof duality \cite{hammer1984roof}, implemented in the Ocean SDK by D-Wave. The method aims to find partial assignment to binary variables in quadratic unconstrained binary optimization (QUBO) formulation, an equivalent form to the Hamiltonian\footnote{D-Wave SDK converts the QUBO formulations to Hamiltonians internally}. Despite its fast  processing time, the method only works in a few special cases that rarely happen in practice, as seen in our comprehensive experiments.  Several other methods also target partial assignment of variables in QUBO \cite{boros2006preprocessing,Rother2007OptimizingBM,lange2019combinatorial}, however, their \emph{high time-complexities} make them unsuitable for QA, in which a high reduction time can nullify the fast processing advantage of QPUs.

To this end, we investigate the task of reducing (final) Hamiltonian to an ``equivalent'' albeit smaller Hamiltonian to save on hardware resource. Given an Hamiltonian $H$ that encodes a classical optimization problem, a reduction of $H$ is a pair of a new Hamiltonian $H_r$ and a mapping $f$ that maps, in a polynomial time, each ground energy state (aka optimal solution) of $H_r$ to a ground energy state of $H$.  Thus, the ground energy state of $H$ that encodes an optimal solution to a optimization problem, can be found by finding those of $H_r$ and performing a mapping with $f$. 
 An effective Hamiltonian reduction that results in small  $H_r$ can lead to a huge saving in physical qubits.

 We introduce a novel notion of \emph{non\nobreakdash-separable~group}, defined as a subset of spins (or logical qubits) in a Hamiltonian that obtains the \emph{same value in ground states}. A group of non-separable spins can be merged into ones, and the weights associated with them can be combined to result in a Hamiltonian with fewer spins. Thus, the identification of non-separable spins lead to natural methods to reduce Hamiltonian.

Through developing theory on non-separable groups, we develop an efficient \underline{Fast} \underline{Ha}miltonian \underline{Re}duction, or FastHare that, iteratively, detects and merges non-separable groups of spins.  It has a provable worst-case time complexity of only $O(\alpha n^2)$, for some user-defined parameter $\alpha$ while exhibiting linear running time in practice. FastHare focuses on identification of small non-separable  groups of size 2 and 3. Further, it utilizes \emph{non-separability index}, a measure on how "non-separable" a group is, of small groups to aid in locating larger non-separable groups.  Our approach is different than the vast majority of existing reduction techniques that rely on identification of partial assignments on variables and has the lowest time-complexity of all.

We perform the first large-scale benchmarks for the feasibility of the reduction on both synthesized Hamiltonian and  3000+ instances from the MQLib library. The roof duality \cite{boros2002pseudo}, implemented in D-Wave's library, \textit{cannot reduce any synthesized instances and  only reduce 8.9\% of MQLib instances}.
In contrast, FastHare  can reduce 100\% of synthesized instances and 43\% of MQLib instances. 
%\mt{ambiguous. We never mentioned that some instances cannot be reduced, thus when we said reduce 43 instances, noone can understand what it means.}
%\td{I provide the numbers for roof duality in D-Wave first, to lower the reader's expectation}
And when it does, it shows a high level of effectiveness with an average 62\% physical qubits saving and 0.3s processing time. Thus, it makes Hamiltonian reduction techniques  \emph{an inexpensive necessity} and ready to be adopted for QA.

%\mt{Remove this listing. It's redundant! Exactly what we said in the above 2 paragraphs already. Didn't provide anything new, nor summarize anything either.}

%Our contribution is summarized as follows
%\begin{itemize}
%	\item We propose and investigate the Hamiltonian reduction problem that works directly on the Hamiltonian (in contrast to QUBO formulation), aiming directly towards QA and Adiabtic Quantum Computing paradigms.
%	\item We introduce the notion of non-separable groups and non-separability index. We also develop the theory on non-separability that guide the construction of reduction algorithms.
%	\item We propose FastHare, the Hamiltonian reduction with a provable low time-complexity of $O(\alpha n^2 )$ and is widely effective for both synthesized and real-world Hamiltonian instances.
%	\item We provide the first large-scale study on the applicability of reduction techniques on Hamiltonian, showing the effectiveness in a large fraction of tested instances and a significant saving on the numbers of both logical and physical qubits.
%\end{itemize}

\textbf{Organization.} We begin by introduce Ising model and prelimnaries in Section~\ref{sec:prelim}. The theory on non-separability and reduction techniques based on identifying non-separable groups are presented in Section~\ref{sec:non-separability}. FastHare is introduced in Section~\ref{sec:fasthare} and the experiments is discussed in Section~\ref{sec:experiments}. Finally, Section~\ref{sec:conclusion} concludes the paper.% with some future research direction.

% !TEX root = main.tex
\section{Preliminaries}
\label{sec:prelim}

% Ising Hamiltonian: QUBO
% Quantum Annealing
% Minor-embedding: logical/physical qubits
% Polynomial-time Hamiltonian Reduction: Reduction ratio

%\begin{itemize}
%	\item Steps to solve an optimization problem on Ising machine: (Pseudo-boolean functions) -> QUBO -> Hamiltonian -> Minor-embedding -> Interfacing and solving problem on Ising machines
%	\item Quantum Annealing and Hamiltonian (Ising formulation)
%	\item QUBO formulation
%	\item Formulation compression
%	\item Weak/strong optimality-preserving mapping between two formulations.	
%\end{itemize}
%

% For a cut $S$, let $Q_s = [q_u]_{u \in V}$, where 
%\begin{equation*}
%q_u = \begin{cases}
%1 & \text{ if } u \in S\\
%-1 & \text{ if } u \in V \setminus S\\
%\end{cases}
%\end{equation*}
%We have, 
%\begin{equation*}
%	C(S) = \frac{1}{4} \sum_{(u,v) \in E} w(u,v) (q_u - q_v)^2
%\end{equation*}
We present Ising Hamiltonian that encodes combinatorial optimization problems and the quantum annealing process to solve the formulated problem on quantum annealers. Further, we define a new notion of \emph{polynomial-time Hamiltonian reduction} and the problem of finding efficient Hamiltonian reduction.
 
\subsection{Ising model and QUBO}
Quantum annealers including D-Wave's can solve optimization problems formulated as an Ising model \cite{lucas2014ising}. The Ising model describes a physical systems with $n$ sites. Each site $i$ is associated with a discrete variable $s_i \in \mathbb{S} = \{-1, +1\}$, representing the site's spin. Each assignment of spin value $\mathbf{s} \in \mathbb{S}^n$, called a \emph{spin configuration}, associates with an energy of the system, defined through the \emph{Ising Hamiltonian}

\begin{align}
	H(\mathbf{s})  =  -\sum_{i=1}^n h_i s_i - \sum_{i,j = 1}^n J_{ij} s_i s_j = -\mathbf{h}^T s - \mathbf{s}^T \mathbf{J} \mathbf{s} 
\label{eq:ising}
\end{align}	
where $h_i$ is the \emph{external magnetic field} at site $i$ and  $J_{ij}$ is the \emph{coupling strength} between sites $i$  and $j$. For a pair $i, j$,  $J_{ij} >0$ ($J_{ij} <0$) indicates a \textit{ferromagnetic} (\textit{antiferromagnetic}) interaction.

The \emph{configuration probability}, the probability that the system is in a state with spin configuration $s$ is given by the Boltzmann distribution with inverse temperature $\beta \geq 0$
\[ 
P_{\beta}( s ) = \frac{e^{-\beta H(s)}}{Z_{\beta}},
\]
where $\beta = (k_BT)^{-1}$, and the normalization constant
\[
Z_{\beta} = \sum_{s \in \mathbb{S}^n} e^{-\beta H(s)}
\]
is the \textit{partition function}.

The ground state of an Hamiltonian associates with the spin configuration of lowest energy
\begin{align}
\mathbf{s}^* = \arg \min_{\mathbf{s} \in \mathbb{S}^n} H(\mathbf{s})  
\label{eq:ising}
\end{align}	
and can be searched for using the quantum annealing process.

\noindent \emph{Quadratic Unconstrained Binary Optimization (QUBO)}. Another popular formulation to encode optimization problem for quantum annealing is QUBO that minimizes a quadratic polynomial over binary variables
$$\mathbf{x}^*=\arg \min_{\mathbf{x} \in \{0, 1\}^n}   Q(\mathbf{x}) = \sum_{i,j \in [n]} q_{ij}x_ix_j,$$
where $\mathbf{x} = (x_1,\cdots,x_n)\in \{0,1\}^n$.

A QUBO can be easily converted back and forth to an Ising Hamiltonian by changing variables $x_i = \frac{s_i+1}{2}$ \cite{choi2008minor}. 
%\begin{align*}
%	H(\mathbf{s}) &= \sum_{i,j \in [n]} q_{ij} \frac{s_i+1}{2} \frac{s_j+1}{2}\\
%	&=-\sum_{i,j \in [n]} J_{ij}s_is_j - \sum_{i} h_i s_i - B,
%\end{align*}
%where
%\begin{equation*}
%\begin{cases}
%	J_{ij} &= -\frac{q_{ij}}{4}, \\
%	h_i &= -\sum_{j\ne i} \frac{q_{ij}}{4} - \frac{q_{ii}}{2},\\
%	B &= -\sum_{i,j \in [n]}\frac{q_{ij}}{4} .
%\end{cases}
%\end{equation*}

\subsection{Quantum Annealing (QA)}
QA \cite{finnila1994quantum, kadowaki1998quantum} is a class of methods to find global optima in combinatorial optimization problems, especially when optimization landscapes are full with local optima. The method is inspired by the classical simulated annealing (SA) method in which an ``annealing schedule'' dictates the temperature variation that in turns decides the probability that a candidate state switch to  neighboring states.

In QA, quantum-mechanical fluctuation such as quantum annealing is utilized to explore the solution space, mimicking the idea of thermal fluctuations in SA. The system evolves from an initial Hamiltonian ground state that is easy to find and setup to a \textit{final Hamiltonian} ground state that encodes the optimization  problem. QA  is closely related to quantum adiabatic evolution, used in adiabatic quantum computation \cite{farhi2000quantum,albash2018adiabatic}, however, the adiabatic conditions are relaxed for faster processing time.  

\emph{Embedding Hamiltonian to QPU Hardware Topology}. Since the qubits in an quantum annealer are not necessarily all-to-all connected, the Ising Hamiltonian for the orginial problem often need to be mapped to a hardware Ising Hamiltonian through a process called \textit{minor embedding} \cite{choi2008minor,choi2011minor}. The process will map each qubit in the original Hamiltonian, termed \emph{logical qubits} to one or multiple \emph{physical qubits} on the annealer. The solution of the embedded Hamiltonian induces the solution to the original Hamiltonian, when sufficiently large coupling strengths are used among physical qubits that associate to the same logical qubit \cite{choi2008minor}.  An example of minor-embedding on the D-Wave annealer can  be seen in Fig.~\ref{fig:compression}.

%
%\begin{figure}
%	\begin{center}
%		\includegraphics[width=0.25\textwidth]{figures/pegasus.pdf}
%		\label{fig:p4}
%	\end{center}
%
%	\caption{\small  $P_3$ Pegasus topology on the latest D-Wave annealer, D-Wave Advantage. Physical qubits  are shown using blue dots and their couplers are displayed as gray lines.}
%	\tnote{Add example for minor-embedding, preferrably on Pegasus}
%	\label{fig:p4}
%
%\end{figure}

\subsection{Polynomial-time Hamiltonian Reduction}
We introduce a new notion of  reduction among Hamiltonians, following the polynomial-time reductions among NP-complete problems \cite{10.5555/574848}.
\begin{definition}[Polynomial-time Hamiltonian Reduction]
	Given two Ising Hamiltonians $H(\mathbf{x})$ and $H'(\mathbf{y})$ with $\mathbf{x} \in \mathbb{S}^n$ and $\mathbf{y} \in \mathbb{S}^{l}$, we say that $H(x)$ is \emph{polynomial-time reducible} to $H'(\mathbf{y})$ if and only if
	\begin{itemize}
		\item \underline{Efficient mapping.} There exists a polynomial-time computable function $f: \mathbb{S}^{l} \rightarrow \mathbb{S}^{n}$, called \emph{reduction function}, that maps each spin configuration $y \in \mathbb{S}^{l} $ to a spin configuration $x \in \mathbb{S}^{n}$.
		\item \underline{Optimality-preserving.} Map each ground state of $H'(y)$ to a ground state of $H(x)$. That is for any
		\[
		\mathbf{y}^* = \arg \min_{\mathbf{s} \in \mathbb{S}^{l}} H'(\mathbf{y}),
		\] 
		we have
		\[
		H\left( \mathbf{x}^*=f(\mathbf{y}^*)\right)= \min _{\mathbf{x} \in \mathbb{S}^{n}} H(\mathbf{x}).
		\]
	\end{itemize}
\end{definition}
We use the notation $H(\mathbf{x} )\xrightarrow{f} H'(\mathbf{y} )$ to denote that $H(\mathbf{x} )$ is polynomial-time reducible to $H'(\mathbf{y})$ with the reduction function $f$. When the context clear, we also use  \emph{Hamiltonian reduction} or  \emph{reduction} in place for polynomial-time Hamiltonian reduction.   

The reduction function $f$ in this paper will be, in most cases, a simple linear map that assigns $x^*_i = {-1}^{sg(i)} y^*_{\pi(i)}, i = 1,\ldots,n$ where $\pi(i) \in \{1,\ldots,l\}$ and $sg(i) \in \{ 0, 1\}$. 
%  
%$H_1 = \min_{\mathbf{s}_1 \in \{-1,+1\}^n} -\mathbf{s}_1 J_1 \mathbf{s}_1^T - \mathbf{s}_1 h_1^T \text{ and } H_2 = \min_{\mathbf{s}_2 \in \{-1,+1\}^m}$ $ -\mathbf{s}_2 J_2 \mathbf{s}_2^T - \mathbf{s}_2 h_2^T$. Consider a  function $M: \mathbf \{-1,+1\}^n \rightarrow  \{-1,+1\}^m$. \fdef{$M$}  Let $\mathbf{S}_1^*$ be the set of optimal solutions of $H_1$. Let $\mathbf{S}_2^* = \{M(\mathbf{x}_1 ^*):\mathbf{x}_1 ^* \in \mathbf{S}_1^*\}$
%
%\begin{definition}
%	We say $M$ is a optimality-preserving reduction from $H_1$ to $H_2$ if all solutions in $\mathbf{S}_2^*$ are optimal solutions of $H_2$.
%\end{definition}
%\begin{itemize}
%	\item 
%	 We say $M$ is a weak optimality-preserving mapping from $H_1$ to $H_2$ if all solutions in $\mathbf{S}_2^*$ are optimal solutions of $H_2$.
%	\item We say $M$ is a strong optimality-preserving mapping from $H_1$ to $H_2$ if $\mathbf{S}_2^*$ is the set of all optimal solutions of $H_2$.
%\end{itemize}

\emph{Composition of reductions.} The composition of two or more Hamiltonian reductions is also a Hamiltonian reduction. Given two Hamiltonian reductions $H_1(\mathbf{x} )\xrightarrow{f_1} H_1(\mathbf{y} )$
and $H_2(\mathbf{y} )\xrightarrow{f_2} H_3(\mathbf{z} )$, we can verify that $H_1(\mathbf{x} )\xrightarrow{f_1\circ f_2} H_3(\mathbf{z} )$, i.e., $H_1(\mathbf{x} )$ is also reducible to $H_3(\mathbf{z} )$ with reduction function  $f_1\circ f_2$.

\emph{Reduction ratio.} Preferably, we want to reduce each Hamiltonian $H(.)$ to a smaller Hamiltonian $H'(.)$. Here, the size of a Hamiltonian $H(.)$, denoted by $size(H)$ can be measured  as either the number of logical qubits, the number of couplings, or the number of physical qubits. The \emph{reduction ratio} of a Hamiltonian reduction is defined as  \begin{align}
	\label{eq:reduction_ratio}1 - \frac{size(H')}{size(H)}.
\end{align} Without otherwise mention, we will measure the \textit{size as the number of physical qubits}  needed to implemented the Hamiltonian on QPU hardware topology, e.g. through minor-embedding.  The maximum reduction ratio is 100\% when $H(.)$ can be reduced to an empty Hamiltonian, i.e., the ground state of $H(.)$ can be found using the reduction function.

\emph{Efficient Hamiltonian reduction problem.} Our main goal is to develop  Hamiltonian reduction algorithms that maximizes the reduction ratio. It is critical that the proposed reduction algorithm has a \textit{low time-complexity} to make sure the reduction time does not dominate the solving time on the quantum annealer. 
% !TEX root = main.tex

\begin{figure}
	\begin{center}
		\includegraphics[width=0.5\textwidth]{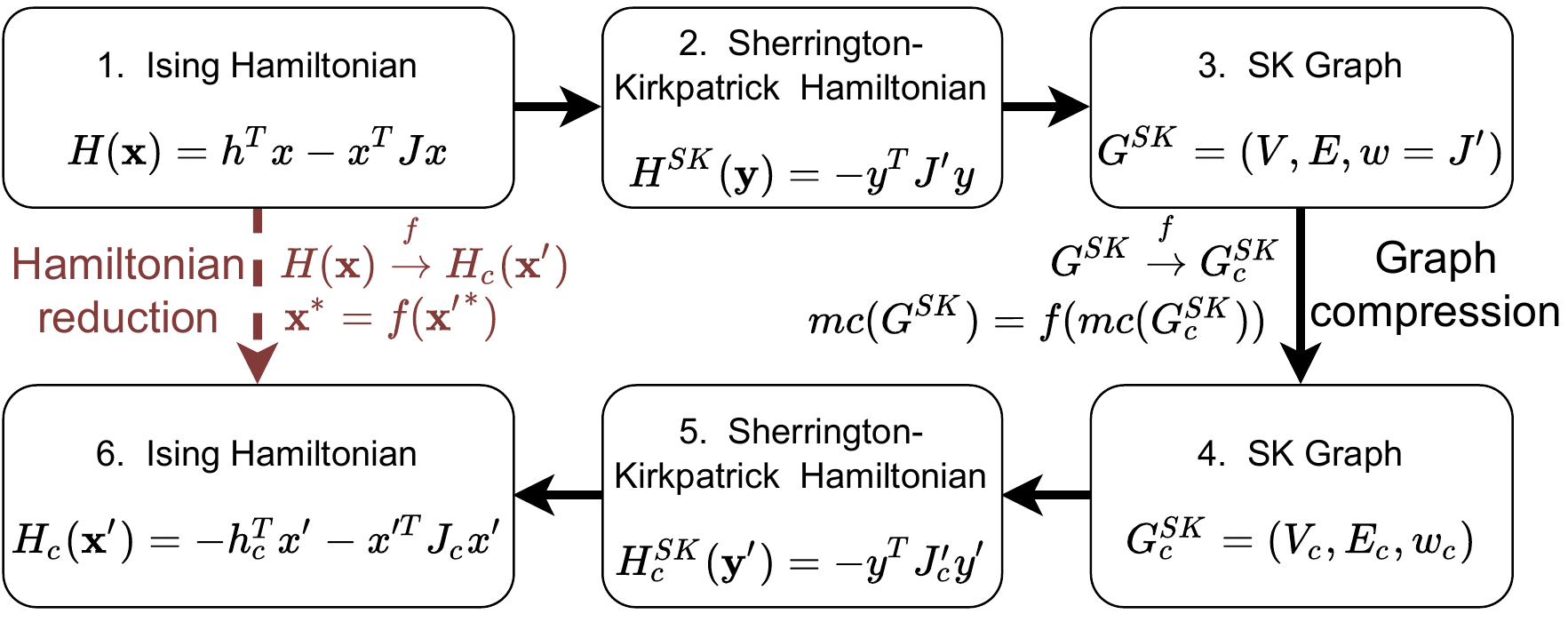}
	\end{center}	
	\caption{\small  Hamiltonian reduction via compressing non-separable groups in SK graph.}	
	\label{fig:framework}	
\vspace{-0.1in}
\end{figure}

\vspace{-0.12in}
\section{Non-separability Theory and\\ Graph-based Hamiltonian Reduction}
\label{sec:non-separability}
%\begin{itemize}
%	\item Intrinsic Hamiltonian Graph and Weighted min-cut: 1) Definitions, examples: *We propose a new ...* 2) Intrinsic: single graph to capture all the parameters of the problem (comparing to the other graph that capture only the coupling) 3) multiple Hamiltonians mapped to the same Intrinsic Hamiltonian Graph (by varying which node node carries the weight for the linear terms) 
%	\item Graph compression via non-separability theory
%    \item Definitions of Weak/strong non-separable groups, examples
%    \item Properties: 1) Hereditary: any subgroup is (weak/strong) non-separable groups (NGs) 2) Closed under intersection (union of two weak/strong NGs is also a NG) 3) Close under union of two intersecting subsets for (only) Strong NS (but NOT weak NG) 4) Non-separability index
%    \item Compression of NGs: 1) compression of a SINGLE weak/strong NG 2) Simultaneous compression of multiple Strong NGs 3) Decompression
%    \item Definition of (strong) antipolar of two strong NGs, examples.
%    \item Finding antipolar groups (AGs) via Node Flipping: find NG on a flipped graph.
%%    Making NS groups  via Node Flipping
%    \item Multi-level Compression Framework:  Repeat if exists multiple STRONG NGs and AGs -> compress, else if exist a weak NG -> compress, else find a group (with high non-separability index) and transform the group into a NG via node flipping. (Coarsening + Uncoarsening: coarsening + flipping)
%\end{itemize}

% Minimum-cut on SK Graphs
% 
%
%
%%

In this section, we propose a Hamiltonian reduction framework via graph compression  as shown in Fig.~\ref{fig:framework}. First, we convert Ising Hamiltonian into Sherrington-Kirkpatrick (SK) Hamiltonian, and then SK graph  that minimum-cut induces the ground state for  the Hamiltonian. We then develop \emph{non-separability theory} for SK  graph and show how compressing non-separable groups in the graph can lead to efficient Hamiltonian reduction.

%we introduce a concept of  \emph{non-separable group} and present a compression framework for non-separable groups.

%\begin{figure}[htp!]
%	\centering  
%	\includegraphics[width=0.45\textwidth]{figures/compression-flow.pdf}   
%	\caption{Hamiltonian compression via non-separability theory. 
%		\label{fig:compression}}  
%\end{figure}
\vspace{-0.3cm}
\subsection{Minimum-cut on Sherrington-Kirkpatrick (SK) Graphs} 
 We introduce a new graph, called \textit{Sherrington-Kirkpatrick (SK) graph}  that encloses \underline{both} the coupling strengths and the external fields in an Ising Hamiltonian. More importantly, \emph{finding the weighted mininmum-cut on the SK graph is equivalent to finding the   ground state of the Ising Hamiltonian}. Thus, the SK graph provides a pure graph theory tool for minimizing the energy of  Ising Hamiltonians.

\paragraph{Construction} Given an Ising Hamiltonian $H(\mathbf{x})= \mathbf{h}^T \mathbf{x} + \mathbf{x}^T \mathbf{J}\mathbf{x}$ with $n$ variables, 
	the SK graph of $H(\mathbf{x})$ is denoted by $G^{SK}_H = (V, E, w)$. The set of nodes $V=\{ 1, 2, \ldots, n, n+1\}$ in which nodes $1, 2,\ldots, n$  correspond to the variables $x_1, x_2, \ldots, x_n$ in the Hamiltonian.  Node $(n+1)$ is added to capture the external fields $\mathbf{h}$. The set of undirected edges $E$ consists of undirected edges $(i, j)$ with weight $w_{ij}= J_{ij} + J_{ji}$ for $1\leq i, j \leq n$ and $(i, n+1)$ with weights $w_{i, n+1}=h_i$. For efficiency, we only retain in $E$ edges with non-zero weights.   
	
	We denote by $\mathbf{J'}$ the weighted adjacency matrix of $G^{SK}$. $\mathbf{J'}$ can be seen as the result of appending the external fields $\mathbf{h}$ to the right of $\mathbf{J}$ (after assigning $J_{ij} = J_{ij} + J_{ji}, J_{ji} =0$ for $i <j$).  For $\mathbf{y} \in \mathbb{S}^{n+1}$,  $\mathbf{J'}$ corresponds to a Hamiltonian
	\[
		H^{SK}(\mathbf{y}) = -\mathbf{y}^T \mathbf{J'} \mathbf{y}.
	\] 
	$H^{SK}$ contains  \emph{no external fields} and is in a form of a Sherrington-Kirkpatrick Hamiltonian \cite{panchenko2013sherrington}, hence, we named the constructed graph SK graph. In fact, we can prove that 
	 \begin{align}
		\nonumber\min_{\mathbf{y}\in \mathbb{S}^{n+1}} H^{SK}(\mathbf{y}) &= \min_{\mathbf{y}\in \mathbb{S}^{n+1}} 
	-\mathbf{y}^T \mathbf{J'} \mathbf{y}\\
	\nonumber&= \min_{\mathbf{y}\in \mathbb{S}^{n+1}} 
	-\sum_{1\leq i, j \leq n} J_{ij} y_i y_j - y_{n+1} \sum_{1\leq i \leq n}  h_{i} y_i\\ 
	\label{eq:sk}&=	\min_{\mathbf{x}\in \mathbb{S}^{n}} H(\mathbf{x}). 	
 	\end{align}	 	
	The last equality holds as we can always replace $\mathbf{y}$ with $-\mathbf{y}$ to ensure $y_{n+1} = 1$ without changing the energy of the Hamiltonian $H^{SK}(\mathbf{y})$.
	
\paragraph{Equivalence between minimizing energy and weighted min-cut (WMC) on SK graph} 
%We establish the equivalence between minimizing energy of a Hamiltonian and a version of WMC on SK graph. 
For any subset $S \subseteq V$,  $S$ induces a cut $\langle S, V\setminus S\rangle$, consisting of the edges crossing $S$ and $V\setminus S$. The  capacity of the cut is defined as 
%\vspace{-0.2cm}
	\[
	c(S) = \sum_{(u,v) \in  \langle S, V \setminus S\rangle} w_{uv}.
  \]
We consider the following variation of  the weighted min-cut (WMC) problem of finding 
$$mc(G) = \arg \min_{S \subseteq V} c(S).$$	
Remark that the cut space includes the empty cut $S = \emptyset$ (or equivalently $S=V$). This is different from the standard minimum-cut problem in which cuts often contain at least one node on each side. For example, since $c(\emptyset) = 0$, it follows that, \[
MC(G) = \min_{S \subseteq V} c(S)  \leq 0,\] 
where $MC(G)$ denotes the minimum capacity of any cut. Thus, min-cuts in WMC often have negative capacities.
	
There is a one-to-one mapping between the capacity of the cut in $G^{SK}$ to the energy of the Hamiltonian $H^{SK}$. Define for a subset $S\subseteq V$, the corresponding vector $\mathbf{y}^{(S)} \in \mathbb{S}^{n+1}$, in which for $v \in V$
\vspace{-0.3cm}
\begin{align*}
	y_{v}^{(S)} = \begin{cases}
		+1 & \text{ if } v \in S, \\
		-1 & \text{ if } v \notin S,.
	\end{cases}
\end{align*}	

We have, 
\begin{align*}
	c(S) &= \sum_{(u,v) \in  \langle S, V \setminus S\rangle} w_{uv}=  \frac{1}{4}\sum_{(u,v) \in E} w_{uv}(y_{u}^{(S)} - y_{v}^{(S)})^2  \\
	&= -\frac{1}{2} \sum_{(u,v) \in E} w_{uv}s_us_v + \frac{1}{2}\sum_{(u,v) \in E} w_{uv}\\
	&= H^{SK}(\mathbf{y}^{(S)}) + c_w,
\end{align*}
where $c_w = \frac{1}{2} \sum_{\left(u,v\right) \in E} w_{uv}=\frac{1}{2} \sum_{i,j} J'_{ij}$ is a fixed value that depends only on $w$.\\

Thus, finding the lowest energy of Hamiltonian $H^{SK}$ and $H(\mathbf{x})$ (from Eq.~\ref{eq:sk}) is the same as finding the WMC on $G^{SK}$. 

\begin{lemma}
	For $c_w = \frac{1}{2} \sum_{i,j} J'_{ij}$,  
\[
	\min_{\mathbf{x}\in \mathbb{S}^{n}} H(\mathbf{x})
	=\min_{\mathbf{y}\in \mathbb{S}^{n+1}} H^{SK}(\mathbf{y})  =	 MC(G) - c_w. 
 \] 
\end{lemma}	

\paragraph{Deriving minimum energy configuration from min-cut} Let $S^* = mc(G^{SK})$ and $\mathbf{x}^{(S^*)}$ be the vector obtained from $\mathbf{y}^{(S^*)}$ by removing the $(n+1)$th element $y^{(S^*)}_{n+1}$. If $y^{(S^*)}_{n+1} = -1$, we  multiply $\mathbf{x}^{(S^*)}$ with $-1$. We can verify that 
\begin{align}
\label{eq:extractx}	H(\mathbf{x}^{(S^*)})= \min_{\mathbf{x}\in \mathbb{S}^{n}} H(\mathbf{x})
\end{align}

\paragraph{SK graph vs. Hamiltonian/QUBO graphs}
The Hamiltonian graph induced by $\mathbf{J}$ does not contain the information on the external fields and, thus, can not represent the Hamiltonian, standing alone.
The QUBO obtained by converting the Hamiltonian to a QUBO formulation has edge weights that are different from the coupling strengths in the hardware. Hence, it may not reflect the physical interactions among the sites.
In contrast, the SK graph  encloses both the external fields and coupling strengths (that are close to the implemented ones on the hardware). It enables the \textit{exploration of the Hamiltonian's energy landscape   via exploring the cut space} on the SK graph.

%\td{Define $f_G()$ the cut function over vector $x_V\in \{-1, +1\}^V$.}

%\vspace{-0.08in}
\subsection{Non-separable Groups (NGs)}
%The concept of  \emph{non-separable group} in WMC problem is extended from concept of \emph{persistency} in QUBO problem \cite{hammer1984roof,boros2006preprocessing}.
%In \cite{hammer1984roof}, we say that for some variable $x_i$ and binary value $b \in \{0,1\}$, the persistency holds if $x_i = b$ for all (or some) optimal solutions. 
%Boros et al. \cite{boros2006preprocessing} introduce the concept of \emph{quadratic persistencies},
%which is extended from persistency for binary relations.
%%, the persistency is extended for binary relations, which . 
%For two variables $x_i,x_j$, we say that  the inequality $x_i \le x_j$ is \emph{quadratic persistent} if $x_i \le x_j$ for all (or some) optimal solutions.
%
%
%In this work, we introduce a concept of  \emph{non-separable group} in WMC problem, that can be used for arbitrary size group. Intuitively, we say a set $X$ is a non-separable group if the set $X$ is not separated in all (or some) optimal cuts. 
We introduce new notions of non-separable groups (NGs) in a weighted undirected graph, non-separability index, and a Hamiltonian reduction framework based on identifying non-separable groups. 

%Here, we present the definition of \emph{non-separable group} in WMC problem. 

Let $G=(V, E, w)$ be a weighted undirected graph, e.g, the SK graph of some Ising Hamiltonian. A subset $X \subseteq V$ is called a \emph{non-separable} group, if \underline{all} min-cuts   on $G$ will have all nodes in $X$ on one side. Here, we use min-cut to refer to an optimal cut for the  WMC problem on $G$. If $X$ stays completely on one side of  some (but not all) min-cuts, we say $X$ is a \emph{weakly non-separable} group. As we will show in the next subsection, all nodes in a (weakly) non-separable group can be merged into a single node, creating a smaller graph. Importantly, any min-cut in the smaller graph  can be easily extended to a min-cut in $G$.

%\begin{definition}[Non-separable group]
%	Consider a graph $G = (V,E,w)$, for any set $X \subseteq V$.
%	\begin{itemize}
%		\item We say the set $X$ is a \emph{weak non-separable group} on the graph $G$ if there exists an optimal solution of WMC on $G$ that does not separate $X$.
%		\item We say the set $X$ is a \emph{strong non-separable group} on the graph $G$ if all optimal solutions of WMC on $G$ do not separate $X$.
%	\end{itemize}
%\end{definition}

\paragraph{Properties of non-separable groups} We show the basic properties of non-separable groups, including hereditary, and the closesure under intersection and union.

\begin{lemma}
	\label{lem:properties}
	Let $X, Y$ be  non-separable groups on $G$. 
	\begin{enumerate}
		\item \underline{Hereditary.} Any subset of $S \subseteq X$ is also non-separable. This statement also holds when $X$ is a weakly non-separable group.
		\item \underline{Closure under intersection and union.} Both $X\cap Y$ and $X\cup Y$ are non-separable.  The statement also holds when only one of $X$ or $Y$ is non-separable and the other is weakly non-separable.		
	\end{enumerate}
\end{lemma}
The proof comes directly from the definition of  non-separable  and weakly non-separable groups. 
%\tnote{proof in the appendix}

%In other words, if $S$ separates $X$, there exists a pair of nodes $u, v \in X$ that are separated by $S$, i.e., $u \in  S$ but $v \notin S$.
%there must be a pair of nodes $u, v \in X$, separated by $S$, i.e., $u \in  S$ but $v \notin S$.

\paragraph{Non-separability index}
We propose a measure, termed \emph{non-separability index}, to quantify how ``difficult'' to separate a group of nodes $X \subseteq V$. 
Here, we say a cut $S\subseteq V$ \emph{separates} a set $X$ if there exist two nodes $u, v \in X$ such that $u \in S$ and $v \notin S$. Formally,  
\begin{definition}[Separation]
	Consider a cut $S \subseteq V$ and a subset $X \subseteq V$, we say $S$ separates $X$, denoted by, $S \sep X$ iff 
	$$ X \cap S \notin \{\emptyset,X\}.$$
\end{definition}
\noindent We also denote by  $sep(X) = \{ S \subseteq V: S \sep X \}$  the collection of all cuts in $G$ that separate $X$.

 The non-separability index of $X$ is defined as \textit{the difference between the  minimum capacities of the cuts in $sep(X)$ and those outside $sep(X)$}. 
\begin{definition}[Non-separability index]	
	Given a graph $G$ and a subset $X \subseteq V$, the non-separability index of $X$ is defined as
	\begin{align}
		\label{eq:non-sep}\nu_G(X) =   \min_{S' \in sep(X)} c(S') -  \min_{S \subset V, S \notin sep(X)} c(S).
	\end{align} 
	\label{def:nsi}
\vspace{-0.13in}
\end{definition}

For a non-separable group $X$, the non-separability index is the minimum increase in the cut capacity to turn some min-cut into a new cut that separates $X$. If  $G$ is a SK graph for some Hamiltonian $H(\mathbf{x})$, the non-separability index of $X$ corresponds to   the \emph{energy gap between the ground state and the next excited state that separates $X$}, i.e.,  having two spins in $X$ with opposite signs.

%Let $\widehat{mc}(G) = \min_{S \subset V : S \ne \emptyset} C(S)$. 
%We have
%$$ mc(G) = \min(0,\widehat{mc}(G)).$$

%\begin{corol}
%	$$\nu(V)=\widehat{mc}(G).$$
%\end{corol}
%\begin{proof}
%	By definition,
%	\begin{align*}
%		\nu(V) &= \min_{S \subseteq V, S \cap V \notin \{\emptyset,V\}} \left(C\left(S\right) - \min\left(C\left(\emptyset\right),C\left(\emptyset\right)\right)\right) \\ 
%		&= \min_{S \subset V, S \cap V \notin \{\emptyset,V\}} C(S)  = \widehat{mc}(G)
%	\end{align*}
%\end{proof}

%\begin{lemma}
%	Given a weighted graph $G$ and any subset $X \subseteq V$. Computing $\nu(X)$ is an NP-hard problem.
%\end{lemma}

The non-separability index $\nu_G(X)$ acts as an indicator on whether  node groups are non-separable. When the context is clear, we omit the  graph $G$ and write $\nu(X)$.
\begin{theorem} [Non-separability conditions]
	Given a group of nodes $X \subseteq V$,
	\begin{itemize}
		\item $X$ is a non-separable group if{f} $\nu(X) > 0$.
		\item $X$ is a \underline{weakly} non-separable group  if{f} $\nu(X) = 0$.	
	\end{itemize}
	\label{theorem:separability}
\end{theorem}
\begin{proof}
	We prove the first statement.  	If $X$ is a non-separable group, it follows that none of the min-cuts can appear  in $sep(X)$. From Eq.~\ref{eq:non-sep}, we have 
\begin{align*}
\nu(X) &= \min_{S' \in sep(X)} c(S') -  \min_{S \subset V, S \notin sep(X)} c(S)\\
&=  \min_{S' \in sep(X)} c(S') -  \min_{S \subseteq V} c(S) > 0. 
\end{align*}
   Vice versa, if $\nu(X) >0$, none of the min-cuts can appear in $sep(X)$ (otherwise $\nu(X) \leq 0$).
   
   Similarly, we can show the second statement by noting  that $\nu(X)=0$ if{f} min-cuts appear both in $sep(X)$ and out of $sep(X)$. 
\end{proof}

\paragraph{Antipolar pair} Consider a special case when $X$ contains a pair of nodes $u$ and $v$.  There are three possible cases for the value of $\nu(X)$: 1) $\nu(X) >0$, $X$ is a non-separable group; 2) $\nu(X) =0$, $X$ is a weakly non-separable group; and 3) $\nu(X) < 0$,  in this case, we say $u$ and $v$ is an \emph{antipolar pair}. For an antipolar pair $u, v$, we have, by Eq.~\ref{eq:non-sep}, all min-cuts must belong to $sep(X)$ (otherwise $\nu(X)\geq 0$). In other words, an antipolar pair always stay in different sides in all min-cuts.

As we will show in next subsection, by negating the weights of all edges incident at $u$ (or $v$), we can turn $u$ and $v$ into  a non-separable pair in the new graph.

\begin{figure}[htp!]
    % \vspace{-0.1in}
	\centering  
	\includegraphics[width=0.5\textwidth]{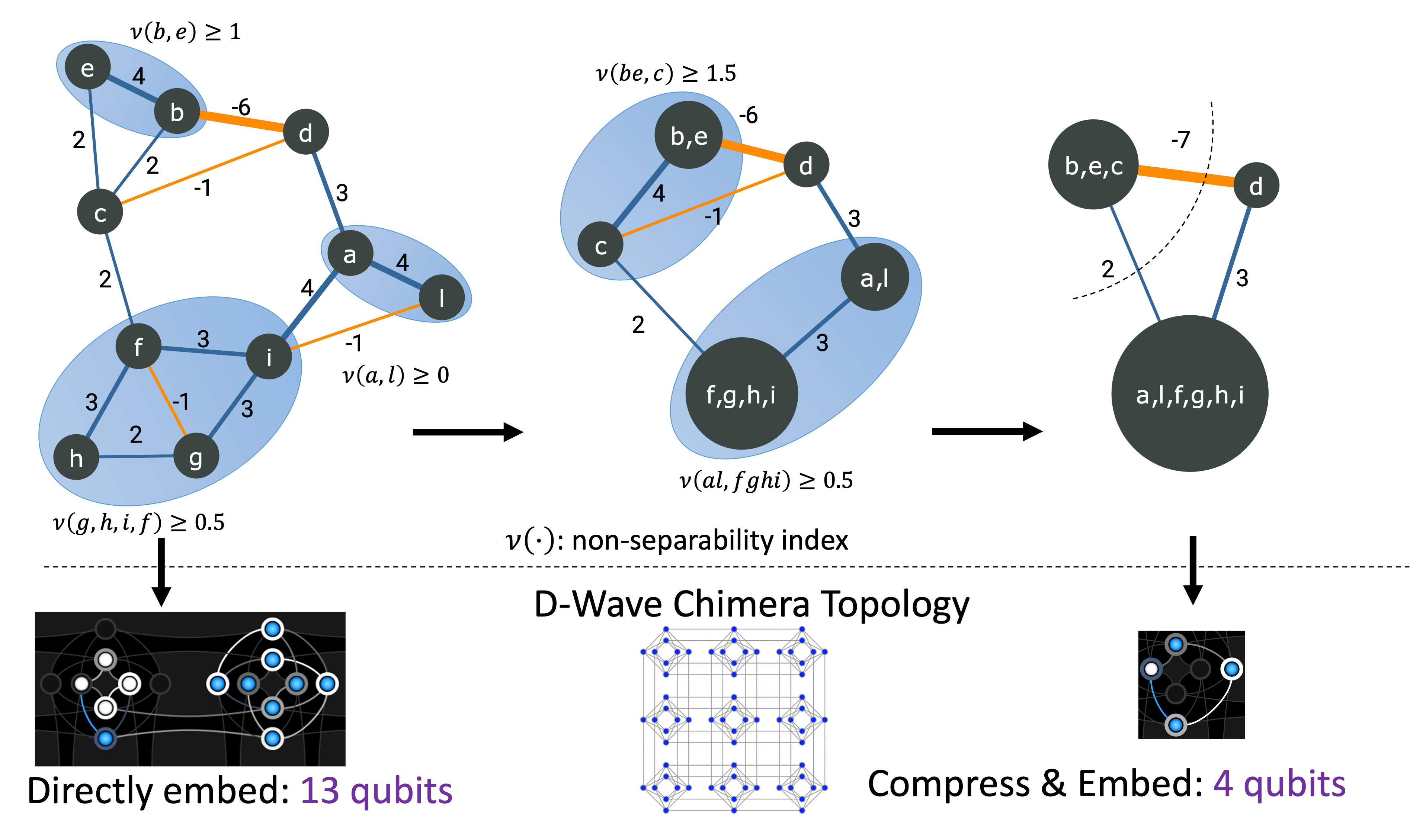}   
	\caption{An example of graph compression framework via non-separability theory. The compression reduce the physical qubits  by 3+ folds (a 69\% reduction ratio).
	}
	\label{fig:compression}  
% 	\vspace{-0.1in}
\end{figure}

% \vspace{-0.1in}
\subsection{Hamiltonian Reduction by Compressing NGs}
%\td{Phase 1, step 1}
As shown in Fig.~\ref{fig:framework}, after converting an Ising Hamiltonian into a  SK graph $G^{SK}=(V, E, w)$, we will  compress non-separable groups (NGs) in $G^{SK}$ into a smaller graph $G^{SK}_c$ that helps us construct the Hamiltonian reduction.

At a high glance, our graph compression framework consists of three steps. First, we identify NGs and antipolar pairs, for example, using methods presented in Section~\ref{sec:fasthare}. Second, we apply the non-separability theory, especially the hereditary and the closure under the union, to enlarge NGs. Finally, we merge each NG into a single node then apply a flip operation to turn antipolar pairs into non-separable pairs that are further merged into single nodes. The three steps are repeated until no  further NGs or antipolar pairs are detected as shown in Fig.~\ref{fig:compression}.   	
\subsubsection{Identification of NGs and antipolar pairs}
In this step, we  search on $G^{SK}$ to identify NGs, weakly NGs, and antipolar pairs, for example, using the algorithm in Section~\ref{sec:fasthare}. We denote by 
 $\mathcal{X}_s$, $\mathcal{X}_w, $ and $\mathcal{R}$ the sets of found NGs, weakly NGs, and antipolar pairs,
 respectively. 
\subsubsection{Enlarging NGs and antipolar pairs}
By applying the closure of NGs and weakly NGs under union, we can enlarge and combine the identified NGs, weakly NGs. Specifically, we apply the following rules:
\begin{itemize}
\item  if $X$ and $Y$ are two NGs, $X \cup Y$ is an NG (Lemma~\ref{lem:properties}).
\item  if $x$ and $y$ is an antipolar pair and $y \in Y$ for some NG $Y$, then for all $z \in Y$, $x, z$  is an antipolar pair.
\item if $x, y$ and $y, z$ are two antipolar pairs, $\{x, z\}$ is an NG.
\end{itemize}
We can use a linear-time algorithm, similar to a node coloring algorithm in a bipartite graph, to repeat the above rules until no further extension is possible. In addition, the above rules can also be extended to include weakly NGs. 
\subsubsection{Compression of NGs and antipolar pairs}
\label{sec:compress}
We present  compression of NGs, weakly NGs, and antipolar pairs. 
In a single round, we will ignore weakly NGs unless there are no NGs nor antipolar pairs.
To preserve min-cuts,  we can  only compress one weakly NG at a time and have to repeat the identification steps. In contrast, multiple NGs and antipolar pairs can be compressed simultaneously in a single round.

\paragraph{Compression of an NG (or weakly NG) to a single node} 
The compression of an NG (or weakly NG) $X$ is done simply by merging nodes in $X$ into a single nodes. Parallel edges will be resolved by aggregating the weights.
 
 \paragraph{Compression of an antipolar pair}
 An antipolar pair $u, v$ is compressed by first, flipping node $u$ (or $v$), followed by merging of $u$ and $v$. The flip of node $u$ is done by negating the weights of all edges incident at $u$. 
 
 Due to the space limit, we omit the proofs on the correctness of the enlarging and compression steps. However, most of the proofs are due to the fact that compression of NGs will preserve min-cuts as each NG will never be separated by any min-cut in the first place. 
\section{Fast Hamiltonian Reduction (FastHare)}
\label{sec:fasthare}
We propose FastHare algorithm, an instance of the compression framework in Section~\ref{sec:non-separability} with the focus on \textit{fast running time}. 
FastHare limits the search to small-size NGs. Further, it uses a nested collection of fast and tight-but-expensive bounds  in scanning for potential NGs.

%We begin with a  lower bound for the non-separability index for groups of any size in Subsection %\ref{subsec:bound}. 
It follows by efficient bounds for small-size NGs of size $2$ and $3$ in Subsection \ref{subsec:fbound}. Third, we present in Subsection \ref{subsec:fasthare}, the efficient search techniques in FastHare that limit the time complexity to $O(\alpha n^2)$ \ref{subsec:fbound}. 
Finally, Subsection \ref{subsec:analysis} provides the complexity analysis.  
%Here, we only attempt to compress the non-separatable groups in which each group consists of two nodes that are the two endpoints of an edge.
% to compress two neighbor nodes.

%We first show a lower bound of non-separability index for groups of two nodes in Subsection \ref{subsec:lowerbound2}. 
%As we described in the recursive framework in Section \ref{sec:non-separability-theory}, we find the the non-separatable groups based on the non-separability index lower. 

%*** Break into two separate bounds
%*** Examples to illustrate which bound is better
%*** Compare to other bounds from reduction techniques for QUBO
%
%
%A QUBO <-> B Hamiltonian
%Two binary variables x, y must have the same value in A iff the corresponding spins x', y' in B must have the same value
% \vspace{-0.1in}
\subsection{Bounds to Prove Non-separability}
\label{subsec:bound}
We begin with a lower bound for the non-separability index for groups of any size. The bound will be used in FastHare to determine whether a group is an NG. 

We define some necessary notations. Given an undirected and weighted graph $G = (V,E, w)$, we extend $w_{uv}$ to define $w_{uv}=0$ if $(u, v) \notin E$. For a node $u \in V$, we denote by $\vwei^{(u)} = (\wei_{u1},\cdots,\wei_{un})$  the weight vector of the node $u$ and by $\|\vwei^{(u)}\| = \sum_{v=1}^{n} |\wei_{uv}|$  the 1-norm of $\vwei^{(u)}$.
We also define $\cuta(S,T)=\sum_{u \in S, v \in T} |\wei_{uv}|$, the total absolute values of weights over all edges between $S$ and $T$. 
\begin{lemma}  [Non-separability index lower bound]
	Consider a graph  $G = (V,E,w)$ and a set $X \subseteq V$, we have
	$$ \nu_G(X) \ge \hat{\nu}_G(X) = \min_{Z \subset X, Z \ne \emptyset} (\cut(Z,X\setminus Z) - P_X(Z)),$$
	%		$$ \nu(X) \ge \nu^*(X) =  \min_{Z \subset X, Z \ne \emptyset} (\cut(Z,X\setminus Z) - P(Z)),$$
	%	/\widehat{MC}(G[X]) - \frac{1}{2}\cuta(X),$$
	where 
	%	\begin{align*}
	%	P(Z)=\min\big(\frac{1}{2}\| \sum_{v \in Z} \vwei_{v,Y} - \sum_{v \in X \setminus Z} \vwei_{v,Y}\|
	%	 \\
	%	 \|\vwei_{Z,V \setminus X}\|,  \|\vwei_{X\setminus Z,V \setminus X}\|\big),
	%	\end{align*}
	%	and 	$Y = \bar{X} = V \setminus X$.
	\vspace{-0.3cm}
	\begin{align*}
	P_X(Z)=\min\big(\frac{1}{2}\sum_{u \in Y}\left|\sum_{v\in Z} \wei_{uv} - \sum_{v\in X \setminus Z} \wei_{uv}\right|, \\
	\cuta(Z,Y),\cuta(X\setminus Z,Y)\big),
	\end{align*}
	and 	$Y = \bar{X} = V \setminus X$.
	
	%	where $G{[X]}$ denotes the induced subgraph by $X$, i.e., $G[X]= (X,E_X = \{(u,v) \in E: u,v \in X\})$.
	\label{lemma:lowerbound}
\end{lemma}
\begin{proof}
	%	Let . 
	Based on Def. \ref{def:nsi}, we have,
	\begin{align*}
	    \nu_G(X) \ge \min_{S \subseteq V, S \sep X} \left(C\left(S\right) - \min\left(C\left(S \setminus X\right),C\left(\bar{S} \setminus X\right)\right)\right)
	\end{align*}
	For any set $S \subseteq V, s.t.,  S \sep X$, let $T = \bar{S} = V \setminus S$. Let $X_S = X \cap S,	X_T = X \cap T$ be the intersections of $X$ and $S,T$, respectively. Let $Y_S = S \setminus X_S, Y_T = T \setminus X_T$ be the intersections of $Y$ and $S,T$, respectively.
	%	\begin{align*}
	%		T = \bar{S} = V \setminus S, \\
	%		X_S &= X \cup S,\\
	%		X_T &= X \cup T,\\
	%		Y &= \bar{X} = V \setminus X, \\
	%		Y_S &= S \setminus X_S, \\
	%		Y_T &= T \setminus X_T.
	%	\end{align*}
	%	\begin{figure}[htp]
	%			\centering
	%		\includegraphics[width=0.6\linewidth]{separation.pdf}
	%	\end{figure}
	
	We have
%	\begin{align*}
%	\cut(S) - \cut(S \setminus X) &= \cut(X_S, T) - \cut(X_S,Y_S) \\
%	%		&= \sum_{(u,v \in E): u \in X-S, v \in T} \wei_{uv} - \sum_{(u,v \in E): u \in X_S, v \in Y_S} \wei_{uv} \\
%	%		&= \sum_{(u,v \in E): u \in X-S, v \in X_T} 
%	&= \cut(X_S,X_T) + \cut(X_S,Y_T) - \cut(X_S,Y_S).
%	%		\\
%	%		&\ge \widehat{MC}(G[X])  - Q(X_S, Y) \\
%	\end{align*}
%	\begin{align*}
%	\cut(S) - \cut(T \setminus X) &= \cut(X_T, S) - \cut(X_T,Y_T) \\
%	&= \cut(X_T,X_S) + \cut(X_T,Y_S) - \cut(X_T,Y_T).
%	\end{align*}
%	Thus, we have,
	\begin{align*}
	&\cut(S) - \min(\cut(S \setminus X),\cut(T \setminus X))\\
	&= \max(\cut(S)-\cut(S \setminus X),\cut(S)-\cut(T \setminus X))\\
	&= \max (\cut(X_S,X_T) + \cut(X_S,Y_T) - \cut(X_S,Y_S), \\
	& \quad \quad \quad \quad \cut(X_T,X_S) + \cut(X_T,Y_S) - \cut(X_T,Y_T))\\
	&= \cut(X_S,X_T) - \min(\cut(X_S,Y_S) - \cut(X_S,Y_T), \\
	&\quad \quad \quad \quad \quad \quad \quad \quad \quad \quad \cut(X_T,Y_T) - \cut(X_T,Y_S))		
	\end{align*}
	
	Let 
	$ Q(S) =  \min(\cut(X_S,Y_S) - \cut(X_S,Y_T), \cut(X_T,Y_T) - \cut(X_T,Y_S)).$
	We have,
	\begin{align*}
	Q(S) &\le \frac{1}{2}(\cut(X_S,Y_S) - \cut(X_S,Y_T) \\
	& \quad \quad+  \cut(X_T,Y_T) - \cut(X_T,Y_S))\\
	%%	&= \frac{1}{2}(\cut(X_S,Y_S) - \cut(X_T,Y_S)) + \frac{1}{2}( \cut(X_T,Y_T) - \cut(X_S,Y_T)) \\
	%	&= \frac{1}{2}\big((\vsum(\vwei_{X_S,Y_S}) - \vsum(\vwei_{X_T,Y_S})) \\
	%	& \quad\quad\quad\quad + (\vsum(\vwei_{X_T,Y_T}) - \vsum(\vwei_{X_S,Y_T}))\big)\\
	%	& = \frac{1}{2}\big(\vsum(\vwei_{X_S,Y_S}-\vwei_{X_T,Y_S})+\vsum(\vwei_{X_T,Y_T} - \vwei_{X_S,Y_T})\big)\\
	%	&\le  \frac{1}{2}\big(\| \sum_{v \in Z} \vwei_{v,Y_S} - \sum_{v \in X \setminus Z} \vwei_{v,Y_S}\| \\
	%	&\quad\quad \quad \quad  + \| \sum_{v \in Z} \vwei_{v,Y_T} - \sum_{v \in X \setminus Z} \vwei_{v,Y_T}\|\big)\\
	%	&= \frac{1}{2}\big(\| \sum_{v \in Z} \vwei_{v,Y} - \sum_{v \in X \setminus Z} \vwei_{v,Y}\| \big)
	%	\\
	&= \frac{1}{2}\sum_{u \in Y_S}\left(\sum_{v\in X_S} \wei_{uv} - \sum_{v\in X_T} \wei_{uv}\right) \\ & \quad \quad  + \frac{1}{2}\sum_{u \in Y_T}\left(\sum_{v\in X_T} \wei_{uv} - \sum_{v\in X_S} \wei_{uv}\right)  \\
	&\le \frac{1}{2}\sum_{u \in Y_S}\left|\sum_{v\in X_S} \wei_{uv} - \sum_{v\in X_T} \wei_{uv}\right|\\
	& \quad\quad + \frac{1}{2}\sum_{u \in Y_T}\left|\sum_{v\in X_S} \wei_{uv} - \sum_{v\in X_T} \wei_{uv}\right|\\
	&\le \frac{1}{2}\sum_{u \in Y}\left|\sum_{v\in X_S} \wei_{uv} - \sum_{v\in X_T} \wei_{uv}\right|
	\end{align*}
	Further, we have,
	
	\begin{align*}
	Q(S) &\le \min(\cuta(X_S,Y_S) + \cuta_(X_S,Y_T), \\ 
	& \quad \quad \quad \quad \cuta(X_T,Y_T) + \cuta(X_T,Y_S)) \\
	&= \min(\cuta(X_S,V \setminus X),  \cuta(X_T,V \setminus X))
	\end{align*}
	Therefore, we have, $Q(S) \le P_X(X_S).$
% 	\begin{align*}
% 	Q(S) \le P_X(X_S).
% % 	&\le \min\big(\frac{1}{2}\sum_{u \in Y}\left|\sum_{v\in X_S} \wei_{uv} - \sum_{v\in X_T} \wei_{uv}\right|, \\ 
% % 	& \quad \quad \quad \quad \quad \cuta(X_S,Y),\cuta(X_T,Y)\big) \\
% % 	&= P_X(X_S).
% 	\end{align*}
	
	Thus, we have,
	\begin{align*}
	\cut(S) - \min(\cut(S \setminus X),\cut(T \setminus X)) &=  \cut(X_S,X_T) - Q(S) \\
	\ge  \cut(X_S,X_T) - P_X(X_S).
	\end{align*}
	
	Hence, we have,
	\begin{align*}
	\nu_G(X) &\ge \min_{S \subseteq V, S \sep X} \left(C\left(S\right) - \min\left(C\left(S \setminus X\right),C\left(T \setminus X\right)\right)\right) \\
	&\ge\min_{S \subseteq V, S \sep X} ( \cut(X_S,X_T) - P_X(X_S))\\
	&= \min_{Z \subset X, Z \ne \emptyset} ( \cut(Z,X\setminus Z) - P_X(Z)).
	\end{align*}
	\vspace{-0.7cm}

	%	Since $S \subseteq V, s.t.,  S \sep X$, we have 
	%	$$ \cut(X_S,X_T) = \cut(X_S, X \setminus X_S) \ge \widehat{MC}(G[X]).$$
	%	Plus, we have,
	%	\begin{align*}
	%		\cut(X_S,Y_T) - \cut(X_S,Y_S) &= \sum_{(u,v \in E): u \in X_S, v \in Y_T} \wei_{uv} - \sum_{(u,v \in E): u \in X_S, v \in Y_S} \wei_{uv} \\
	%		&\ge -\sum_{(u,v \in E): u \in X_S, v \in Y_T} |\wei_{uv}| - \sum_{(u,v \in E): u \in X_S, v \in Y_S} |\wei_{uv}|\\
	%		&= -\sum_{(u,v \in E): u \in X_S, v \in Y} |\wei_{uv}| = -\cuta(X_S,Y).
	%	\end{align*}
	%	Hence, we have
	%	$$ \cut(S) - \cut(S \setminus X) \ge \widehat{MC}(G[X])- \cuta(X_S,Y).$$
	%	Similarly, we have,
	%	\begin{align*}
	%		\cut(S) - \cut(T \setminus X) \ge   \widehat{MC}(G[X])  - \cuta(X_T, Y)
	%	\end{align*}
	%	Thus, we have, 
	%	\begin{align*}
	%		\cut(S) - \min(\cut(S \setminus X),\cut(T \setminus X)) &= \max(\cut(S)-\cut(S \setminus X),\cut(S)-\cut(T \setminus X))\\
	%		&\ge \max(\widehat{MC}(G[X])-\cuta(X_S, Y),\widehat{MC}(G[X])-\cuta(X_T, Y)) \\
	%		& = \widehat{MC}(G[X]) - \min (\cuta(X_S, Y),\cuta(X_T, Y))\\
	%		&\ge \widehat{MC}(G[X]) -\frac{1}{2} (\cuta(X_S, Y)+\cuta(X_T, Y)) \\
	%		&= \widehat{MC}(G[X]) -\frac{1}{2} \cuta(X, Y) \\
	%		&= \widehat{MC}(G[X]) - \frac{1}{2} A(X)
	%	\end{align*}
\end{proof}

% \vspace{-0.1in}
\subsection{Efficient search for NGs}
\label{subsec:fbound}
Now, we use the non-separability index lower bound in Lemma \ref{lemma:lowerbound} to search for non-separable and antipolar pairs of sizes $2$ and $3$. 

\smallskip \noindent \emph{Non-separable pair identification.} Consider an edges $(u,v) \in E$. Our goal is to determine the relation between $u$ and $v$, whether they make an NG,  a weakly NG, or an antipolar pair.
For an edge $(u, v)\in E$, we define \emph{fast score} $\fasts$ and \emph{similarity score} $\sims$ for $(u, v)$ as follows
\begin{equation}
\fasts(u,v) = 2|\wei_{uv}| - \min(\|\vwei^{(u)}\|,\|\vwei^{(v)}\|),
\label{eq:fasts}
\end{equation}
\begin{equation}
\sims(u,v) = 
\begin{cases}
2|\wei_{uv}| - \frac{1}{2} \|\vwei^{(u)} - \vwei^{(v)}\| & \text{ if }	 \wei_{uv} \ge 0,\\
2|\wei_{uv}| - \frac{1}{2} \|\vwei^{(u)} + \vwei^{(v)}\| & \text{ if }	 \wei_{uv} < 0.\\
\end{cases}
\label{eq:sim}
\end{equation}

\begin{lemma}
	Consider a graph  $G = (V,E,w)$. For any edges $(u,v) \in E$, we have:
	\begin{itemize}
		\item If $\max(\fasts(u,v),\sims(u,v)) > 0$,
		\begin{itemize}
			\item if $\wei_{uv} \ge 0$, $\{u,v\}$ is an  NG,
			\item if $\wei_{uv} < 0$, $(u,v)$ is  an antipolar pair.
		\end{itemize}
		\item If $\max(\fasts(u,v),\sims(u,v)) = 0$ and $\wei_{uv} \ge 0$, $\{u,v\}$ is classified as a weakly NG\footnote{$\{u, v\}$ could actually be an NG but the bound is not tight enough to detect}.
	\end{itemize}
\label{lemma:pair}
\end{lemma}
\vspace{-0.2cm}
\begin{proof}
	Let $X = \{u,v\}$ and $Y = V \setminus X$. We consider two cases of $\wei_{uv}$ as follows.
	
	\noindent \emph{Case 1:} $\wei_{uv} \ge 0$.
	Based on Lemma \ref{lemma:lowerbound}, we have, 
	\begin{align*}
		\nu_G(X) &\ge \wei_{uv} - \min(\frac{1}{2}\sum_{z \in Y}|\wei_{uz}-\wei_{vz}|, \\
		&  \quad \quad \quad \quad \quad \quad\cuta(\{u\},Y), \cuta(\{v\},Y))\\
	&	= 2\wei_{uv} - \min(\frac{1}{2} \|\vwei^{(u)} - \vwei^{(v)}\|,\|\vwei^{(u)}\|,\|\vwei^{(v)}\|) \\
		&= \max(\fasts(u,v),\sims(u,v))
	\end{align*}
	Thus, we have:
	\begin{itemize}
		\item If $\max(\fasts(u,v),\sims(u,v)) > 0$, $\{u,v\}$ is an  NG.
		\item If $\max(\fasts(u,v),\sims(u,v)) = 0$, $\{u,v\}$ is classified as a weakly NG.
	\end{itemize}
	\noindent \emph{Case 2:} $\wei_{uv} < 0$. Let $G' = \flip(G,v)$. Similar to Case 1, we have,
	\begin{align*}
	\nu_{G'}(X) \ge \max(\fasts(u,v),\sims(u,v)).
	\end{align*}
	Thus, if $\max(\fasts(u,v),\sims(u,v)) > 0$, $\{u,v\}$ is an  NG in $G'$. In other words, $(u,v)$ is  an antipolar pair in $G$. 
\end{proof}

\smallskip \noindent \emph{Non-separable triple identification.} Consider a group of three nodes $X = \{u,v,z\}$ that has at least 2 edges among the nodes. Apply the lower bound on the non-separability index $\hat{\nu}_G(X)$ in Lemma \ref{lemma:lowerbound} on $X$, we have
%we use  Lemma \ref{lemma:lowerbound}. 
%Based on Lemma \ref{lemma:lowerbound}, we find the . We have,
\begin{align*}
	\hat{\nu}_G(X) < \min_{Z \subset X, Z \ne \emptyset} (\cut(Z,X\setminus Z)) \le MC(G[X]),
\end{align*}
where $G[X]$ is the subgraph  induced by $X$ in  $G$.

\begin{figure}[!ht]
% \vspace{-0.15in}
	\centering 	
	\begin{subfigure}{0.22\textwidth}
		\includegraphics[width=\linewidth]{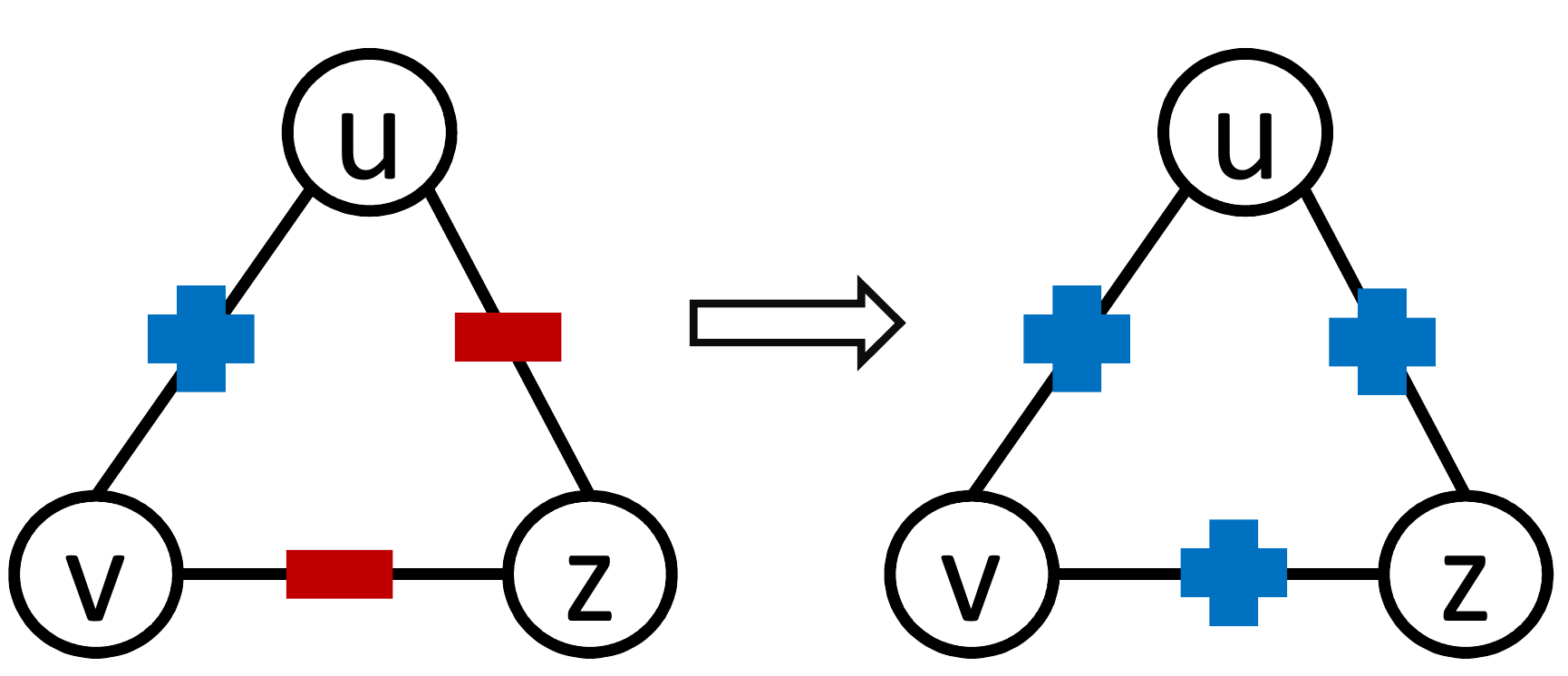}
		\caption{The number of edges with negative weight is even. All edges have non-negative weights after flipping.} 
	\end{subfigure}	
\hfill
	\begin{subfigure}{0.22\textwidth}
		\includegraphics[width=\linewidth]{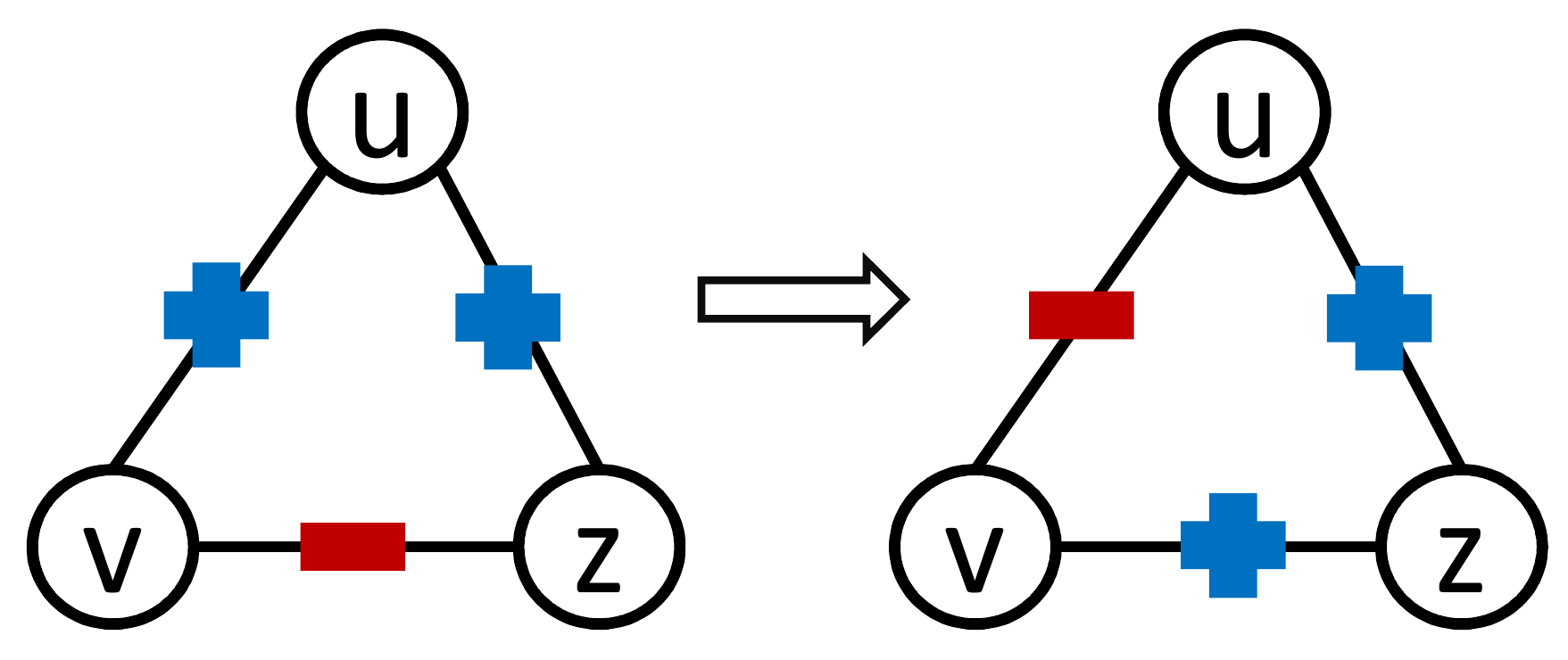}
		\caption{The number of negative weights is odd. The edge $(u,v)$ (with the smallest absolute weight) has a negative weight after flipping.} 
	\end{subfigure}

	\caption{ 
		Flipping nodes in $X = \{u,v,z\}$ to ensure the WMC on the induced graph on $X$ is non-negative. 
		\label{fig:flip}} 
 		\vspace{-0.2in}
\end{figure}

Recall that, we can only find the relation among the nodes in $X$ if the $\hat{\nu}_G(X) \ge 0$. Hence, we will flip the nodes in $X$ such that the WMC in the  subgraph induced by $X$ is non-negative. As every time we flip a node in $X$, we always change the signs of two edges in $G[X]$, the parity on the number of negative weight edges remain the same. Thus, we consider two cases based on the number of edges with negative weight (see Fig. \ref{fig:flip}).
\begin{itemize}
	\item \emph{Case 1:} The number of edges with negative weight is even. In this case, we can flip the nodes in $X$ such that all edges in $G[X]$ is non-negative. Thus, the WMC of $G[X]$ is non-negative.
	\item \emph{Case 2:} The number of edges with negative weight is even. In this case, we can flip the nodes in $X$ such that only the edge, with the smallest absolute weight, has a negative weight after flipping. Now, the WMC of $G[X]$ is also non-negative.
\end{itemize}

Let $\bar{X} \subseteq X$ be the set of nodes that we need to flip so that the WMC of $G[X]$ is also non-negative. Let $G' = \flip_{\tilde{X}}(G)$ and $\tilde{\wei}_{uv},\tilde{\wei}_{uz},\tilde{\wei}_{vz}$ be the weight of $(u,v)$, $(u,z)$, $(v,z)$, respectively, on $G'$. We define the \emph{triangle score} $\tris$ as follows.
\begin{align}
\nonumber\tris(X) &= \min_{x \in X}\big(\sum_{y \in X\setminus\{x\}} \tilde{\wei}_{xy} - \min\big(\|\vwei^{(x)}\|-\sum_{y \in X\setminus\{x\}} |\wei_{xy}|,\\
& \quad\quad \quad\sum_{y \in X \setminus\{x\}} \big(\|\vwei^{(y)}\|-\sum_{z \in X\setminus\{y\}} |\wei_{yz}| \big)\big)\big)
\label{eq:tris}
\end{align}

\begin{lemma}
	Consider a graph  $G = (V,E,w)$ and any set $X \subseteq V$ of size $3$. Let $\bar{X} \subseteq X$ be the set of nodes that we need to flip so that the WMC of $G[X]$ is also non-negative. If $\tris(X) > 0$, we have:
	\begin{itemize}
		\item $\tilde{X}$ and $X \setminus \tilde{X}$ are NG groups.
		\item $\forall u \in \tilde{X}, v \in X \setminus \tilde{X}$, $(u,v)$ is an antipolar pair.
	\end{itemize} 
\label{lemma:tri}
\end{lemma}
\begin{proof}
	Let $G' = \flip_{\tilde{X}}(G)$ and $\tilde{\wei}_{uv},\tilde{\wei}_{uz},\tilde{\wei}_{vz}$ be the weight of $(u,v)$, $(u,z)$, $(v,z)$, respectively, on $G'$.
Let $Y = V \setminus X$, we have, 
\begin{align*}
	\hat{\nu}_{G'}(X) &\ge \min_{x \in X} \big(\sum_{y \in X\setminus\{x\}} \tilde{\wei}_{xy}-\min(\cuta(\{x\},Y),\\
	&\quad\quad\quad\quad\quad\cuta(X\setminus \{x\},Y)) \big) \\
	&= \min_{x \in X}\big(\sum_{y \in X\setminus\{x\}} \tilde{\wei}_{xy} - \min\big(\|\vwei^{(x)}\|-\sum_{y \in X\setminus\{x\}} |\wei_{xy}|,\\
	& \quad\quad \quad\sum_{y \in X \setminus\{x\}} \big(\|\vwei^{(y)}\|-\sum_{z \in X\setminus\{y\}} |\wei_{yz}| \big)\big)\big)\\
	&= \tris(X).
\end{align*}

\noindent If $\tris(X) > 0$, $X$ is an NG on $G'$. Thus, $\tilde{X}$ and $X \setminus \tilde{X}$ are NGs. And $\forall u \in \tilde{X}, v \in X \setminus \tilde{X}$, $(u,v)$ is an antipolar pair.
\end{proof}

% \vspace{-0.12in}
\subsection{FastHare algorithm}
\label{subsec:fasthare}

We now describe the FastHare algorithm to reduce the Hamiltonian. The  algorithm follows the compression framework (see Fig~\ref{fig:framework}) in Section~\ref{sec:non-separability}. It transforms the Hamiltonian reduction task into a graph compression problem. Its main algorithm also consists of multiple rounds, each round consists of three steps: 1) identification of NGs, weakly NGs, and antipolar pairs 2) enlarging step, and 3) compression step.

The identification of NGs is done by computing fast scores, the similarity scores, and the triangle scores for groups of 2 and 3 nodes in the graph. The main trick is to levarage fast score, that can be computed and maintained efficiently after merging and flipping,  to guide the search for potential edges and triangles and attemp to prove their non-separability with more expensive bounds/scores.
The pseudocode of the FastHare algorithm is given in Algorithm \ref{alg:fasthare}.
%, an application of the compression framework. 

%The FastHare algorithm can be view as the Phase 1 of the compression framework in Algorithm \ref{alg:framework}.
%\pnote{Algorithm FastHare is fore phase 1 in the compression framework.}

\begin{algorithm}
	\small 
	\SetKwInOut{Input}{Input}
	\SetKwInOut{Output}{Output}	
	\SetKwFor{Event}{}{}{}
	\Input{A graph $G = (V,E,w)$ and a parameter $\alpha$}
	\Output{A compressed graph.}
	Compute the fast score $\fasts(u,v) \forall (u,v) \in E$ \
	Add top $n \alpha$ edges with the highest fast score to a list $L$ \\
	Compute the similarity score for all edges in $L$\\
	For $(u, v) \in L$ and $w \in \adj(u) \cup \adj(v)$, compute   $\tris(\{u,v,z\}$\\
%	$\ell = 0$, $G_0 = G$\\
	\Repeat{		$\mathcal{X}_s, \mathcal{X}_w, \mathcal{R} = \emptyset$	}
% 	\Repeat{$S = \emptyset$	}
	{
	    	Obtain $\mathcal{X}_s, \mathcal{X}_w, \mathcal{R}$ from pairs and triples with non-negative updated scores (Lemmas \ref{lemma:pair} and \ref{lemma:tri})\\
%		$\mathcal{X}_s = \mathcal{X}_w = \mathcal{R} = \emptyset$\\
% 		$S = \emptyset$\\
% 		\For{$(u,v) \in E: \fasts(u,v) \ge 0$ (see Eq.\ref{eq:fasts})}
% 		{
% 			Add $\{u,v\}$ to $S$
% %			Update 	$\mathcal{X}_s, \mathcal{X}_w, \mathcal{R}$ based on Lemma \ref{lemma:neighbor}
% 		}
% 		\For{$(u,v) \in L$: $\sims(u,v) \ge 0$ (see Eq.\ref{eq:sim})}
% 		{
% 			Add $\{u,v\}$ to $S$
% %			Update 	$\mathcal{X}_s, \mathcal{X}_w, \mathcal{R}$ based on Lemma \ref{lemma:neighbor}
% 		}
% 		\For{$(u,v) \in L$}
% 		{
% 			\For{$z \in \adj_u \cup \adj_v$}
% 			{
% 				\If{$\tris(\{u,v,z\} > 0$)}{	Add $\{u,v,z\}$ to $S$}
% 			}
% 		}
		
%		$\ell = \ell +1 $\\
%		$(G_{\ell},\OPM_\ell) = \compress( G_{\ell-1},\mathcal{X}_s, \mathcal{X}_w, \mathcal{R})$\\
		Compress the graph $G$ based on the list  $\mathcal{X}_s, \mathcal{X}_w, \mathcal{R}$ using the compression in Subsection \ref{sec:compress}\\
		Update the scores and the list $L$ on the new graph\\
	
	}
	Return $G$
	\caption{Algorithm FastHare.}
	\label{alg:fasthare}
%  	\vspace{-0.05in}
\end{algorithm}

%\pnote{
%	%	Pre-compute and update the score.
%	\begin{itemize}
%		\item For each node $v \in V$, maintain the value $A_v = \|\vwei^{(v)}\|$\\
%		\item For each edge $(u,v) \in L$, maintain 
%		$$B_{uv} =\begin{cases}
%		\sum_{z \in \adj_u \cap \adj_v} |\wei_{uz}-\wei_{vz}| - |\wei_{uz}| - |\wei_{vz}| & \text{ if } \wei_{uv} \ge 0\\
%		\sum_{z \in \adj_u \cap \adj_v} |\wei_{uz}+\wei_{vz}| - |\wei_{uz}| - |\wei_{vz}|& \text{ if } \wei_{uv} < 0			
%		\end{cases}$$
%		\item We have $$\fasts(u,v) = 2|\wei_{uv}| - \min(A_u,A_v)$$
%		$$\sims(u,v) = 2|\wei_{uv}| - \frac{1}{2}(A_u+A_v+B_{uv})$$
%		
%		\item Update A,B when flip
%		\item Update A,B when merging two nodes 
%	\end{itemize}
%	
%}
\paragraph{Initialization (Lines 1-3, Alg.~\ref{alg:fasthare})} The FastHare algorithm starts with an initialization phase, followed by a loop of iterations to reduce the Hamiltonian. 
In the initialization phase, we compute the fast score (Eq. \ref{eq:fasts}) for all edges and select the top $n \alpha$ edges with the highest fast score to a list $L$. Then, we compute the similarity score for all edges in $L$ and the triangle score for the groups that have at least one edges in $L$. 

\paragraph{Iterative compression (Lines 5-7, Alg.~\ref{alg:fasthare})} 
In each iteration, we obtain the collection of NGs $\mathcal{X}_s$, the collection of weakly NGs $\mathcal{X}_w$, and the collection of antipolar pairs $\mathcal{R}$ from pairs and triples with non-negative updated scores (Lemmas \ref{lemma:pair} and \ref{lemma:tri}). The scores are computed in the previous iteration (or the initialization for the first iteration). 
Then, we compress the graph $G$ based on the list  $\mathcal{X}_s, \mathcal{X}_w, \mathcal{R}$ using the compression in Subsection \ref{sec:compress}. Finally, we update the scores and the list $L$ on the new graph.

\noindent \emph{Efficiently maintaining the score.} 
For each node $v \in V$, we maintain a value 
$A_v = \|\vwei^{(v)}\|.$ Plus, for each edge $(u,v) \in L$, we maintain a value $B_{uv} = B'_{uv} - |\wei_{uz}| - |\wei_{vz}|$, where
$$B'_{uv} =\begin{cases}
\sum_{z \in \adj_u \cap \adj_v} |\wei_{uz}-\wei_{vz}| & \text{ if } \wei_{uv} \ge 0,\\
\sum_{z \in \adj_u \cap \adj_v} |\wei_{uz}+\wei_{vz}| & \text{ if } \wei_{uv} < 0,		
\end{cases}$$
where $\adj_v$ is the set of neighbors of the node $v$. 

%We can compute the scores based on $A$ and $B$ as follows. 
\noindent For any pair $(u,v) \in E$, we can compute the fast score 
% in $O(1)$, i.e.,
$$ \fasts(u,v) = 2|\wei_{uv}| - \min(A_u,A_v).$$
For any pair $(u,v) \in L$, we can compute the similarity score 
% in $O(1)$, i.e.,
$$\sims(u,v) = 2|\wei_{uv}| - \frac{1}{2}(A_u+A_v+B_{uv}).$$
The triangle score of a group $X$ can also be computed based on the values of $A$.
% For each group $X$ of size three, we can compute the triangle score
% of $X$ in $O(1)$, i.e.,
% \begin{align*}
% \tris(X) &= \min_{x \in X}\big(\sum_{y \in X\setminus\{x\}} \tilde{\wei}_{xy} - \min\big(A_x-\sum_{y \in X\setminus\{x\}} |\wei_{xy}|,\\
% & \quad\quad \quad\sum_{y \in X \setminus\{x\}} \big(A_y-\sum_{z \in X\setminus\{y\}} |\wei_{yz}| \big)\big)\big).
% \end{align*}

% For the similarity score, we only check for the edges in $L$ since we cannot compute the similarity score of other edges in $O(1)$. 
% In each iteration, we check the fast score for all edges in $E$. 
% For the triangle score, even though we can compute the score for all groups in O(1), we limit the search space to groups that have at least one edges in $L$. We remark that,  for a group in which all nodes do not have a connection to a newly merged node in the previous iteration, the scores of that group remain the same as the previous iteration. Thus we do not check the scores of those groups. This helps to reduce the number of groups we check in all iterations to $O(n^2)$.
%the score of a group that
%for any node $v$ that does not connect with a newly merged node in the previous iteration.

\noindent \emph{Updating the scores after flipping a node.} After flipping a node $u$, $\forall v \in V$ the value $A_v$ does not change. We only need to update the value of $B_{uv}$ for all $(u,v) \in L$ such that $v \in \adj_u$. For the edge $(v,z) \in L$ such that $v,z \in \adj_u$, the value of $B_{uv}$ does not changed since the sign of both $\wei_{uv}$ and $\wei_{uz}$ are changed. 
% This can be done in $O(n\alpha)$.
%For the values of $B$, we only need to update 

\noindent \emph{Updating the scores after merging two nodes.} After merging two nodes $(x,y)$ to a new node $z$. We compute the new  value of $A_z$ and update the value $A_u$ for all $u \in \adj_x \cup \adj_y$. 
% This can be done in $O(n)$. 
For the similarity score, we remove all edges in $L$ that one endpoint is $x$ or $y$. Then we update  $B_{uv}$ for all edges $(u,v) \in L$ such that both $u$ and $v$ are adjacent to $x$ or $y$. 
% This can be done in $O(|L|) = O(n \alpha)$. 
We also add at most $\alpha$ edges from $z$ with the highest fast score to $L$ and update the value $B$ of those edges.  This limitation on the number of updated edges is important to keep the running time bounded by $O(\alpha n^2$). 

\subsection{Complexity analysis}
\label{subsec:analysis}

\begin{lemma}
	The time complexity of the FastHare algorithm (Algorithm \ref{alg:fasthare}) is $O(n^2\alpha)$
\end{lemma}
\begin{proof}
For the initialization, the time complexity to compute the fast score, the similarity score, and the triangle score is $O(n^2)$, $O(n^2 \alpha)$, and $O(n^2 \alpha)$, respectively.

For the iterative compression, after flipping a node $u$, we only need to compute for all $(u,v) \in L$ such that $v \in \adj_u$. Thus,
the cost to update the scores  after a flipping is $O(n\alpha)$. 
After the merging two nodes $(x,y)$ to a node $z$, we update the score for all edges $(u,v) \in L$ such that both $u$ and $v$ are adjacent to $x$ or $y$ and the top $\alpha$ edges from $z$ with the highest fast score.
% This can be done in $O(|L|) = O(n \alpha)$. 
% We also add at most $\alpha$ edges from $z$ with the highest fast score to $L$ and update the value $B$ of those edges.  
Thus, the cost to update after a merging is $O(n\alpha)$. In FastHare the total number of flipping/merging is $n$. 
Thus, the total time complexity to update the score is $O(n^2 \alpha)$. 
Plus, in each iteration, we only check for the pairs and triplets that have the scores updated. Thus, the time complexity to check the pairs and triplets is $O(n^2 \alpha)$.

Therefore, in total, the time complexity of the FastHare algorithm is $O(n^2 \alpha)$.
\end{proof}
%\begin{lemma}
%	Let $G'$ be the compressed graph when the Alg. \ref{alg:compress2} is terminated.
%	The complexity of Alg. \ref{alg:compress2} is $O((|V| - |V'|) \sum_{v \in V} d_v^2)$, where $d_v$ is the degree of the node $v$. 
%\end{lemma}
% !TEX root = main.tex
% \vspace{-0.4cm}
\section{Experiment}
\label{sec:experiments}

We perform numerical experiments to assess the performance of the proposed methods in terms of reduction ratio and processing time. Further, we analyze characteristics of the benchmarked instances to identify factors that are important for the reducibility of Hamiltonian. 
\vspace{-0.08in}
\subsection{Experiments settings}

\noindent \textbf{Algorithms.} We compare FastHare algorithm, that is described in Section \ref{sec:fasthare} with
the implementation of D-Wave's SDK\footnote{\url{https://docs.ocean.dwavesys.com/en/latest/docs_preprocessing/reference/lower_bounds.html}}. The implementation of D-Wave's SDK applies a roof duality technique \cite{hammer1984roof} to minimizing assignments for some of the variables \cite{boros2002pseudo,boros2006preprocessing}. 
For the FastHare algorithm, we set the parameter $\alpha = 2$. 
% QPBO and QPBOP that are introduced in \cite{boros2006preprocessing} and implement in \cite{rother2007optimizing}. Here, QPBO uses roof duality technique to compress the instances. QPBOP is an extended version of QPBO using a probing technique. 

\noindent \textbf{Instances.} We benchmark the algorithms on  both synthetic instances and  3000+ instances derived from the popular MQLib collection \cite{DunningEtAl2018}.
\begin{itemize}
	\item \emph{Synthetic instances.} We generate a random network and assign uniformly random weights in some interval to all edges. Here, the random network is generated using Erdos-Renyi (ER) network model (each edge has a fixed probability of being present or absent) and scale-free (SF) network model (the networks whose degree distribution follows a power law) using networkX library \cite{schult2008exploring}. 
	We set the number of nodes to $10,000$ and the average degree to $6$, and  integral weights are uniformly chosen in $\{-2^{10}..2^{10}\}$.
	\item \emph{MQLib \cite{DunningEtAl2018}.} We also benchmark the algorithm over $3,000+$ instances that provided in \cite{DunningEtAl2018}. 
	MQLib is a standard instance library for Max-Cut and QUBO.  All QUBO instances were converted to Max-Cut instances.	The authors collect the data from multiple sources such that Gset \cite{gset}, Beasley\cite{beasley1990or}. They also generated a number of random graphs using Culberson random graph generators \cite{culberson1995hiding} and convert image segmentation problems to Max-cut using the techniques in \cite{sousa2013estimation}.
%	Those instance contains multiple 
%	consider the following real-world application.
%	\begin{itemize}
%		\item \emph{Max-cut.} We take the max-cut instances to solve the Ising spin glass problem \cite{liers2004contributions,liers2004computing}.
%		\item \emph{Modularity.} We consider modularity-maximizing bipartition problem on 6 real-world instances \cite{zachary1977information,knuth1993stanford,lusseau2003bottlenose}.
%	\end{itemize}
\end{itemize}

\noindent \textbf{Metrics.} We compare the performance of the algorithms based on the following metrics.

\emph{Processing time.} We measure the time to reduce the size the instances of algorithms. We exclude any time to read or write data from hard drives.

\emph{Reduction ratio.} We compute the reduction ratio (Eq.~\ref{eq:reduction_ratio}) with size of Hamiltonian  measured in both the number of physical and logical qubits (the number of variables).
% More concretely, let $n_b, n_a$ be the number of physical (logical) qubits that are needed to solve an instance before and after reduction, respectively.
%The reduction ratio is computed as $\left(1 - \frac{n_a}{n_b}\right) \times 100\%$. 
%\td{Check the definition, so that compress to one node $\rightarrow$ 100\%}

%\begin{wrapfigure}{r}{0.25\textwidth}
%	\begin{center}
%		\includegraphics[width=0.2\textwidth]{figures/pegasus.pdf}
%	\end{center}
%	\caption{\small$P_3$ Pegasus topology,  with qubits represented as blue dots and couplers as gray lines. \label{fig:p4}}
%\end{wrapfigure}
We measure the number of physical qubits on the D-wave Advantage QPU's topology, called Pegasus \cite{boothby2020next}. A Pegasus topology $P_M$ contains $8(3M-1)(M-1)$ qubits\footnote{To be precise, $P_M$ contains $24M(M-1)$ qubits. 
However, $8(M-1)$ qubits are disconnected to the remaining.}, in which each qubit connect to at most $15$ others.
% Fig. \ref{fig:p4} shows the $P_3$ Pegasus topology with $128$ qubits. 
The current  Advantage QPU is built on a $P_{16}$ Pegasus topology with $5,640$ qubits.

%In this work, we embed the instances to a $P_{100}$ Pegasus topology with $236,808$ qubits. 

In this work, we use a method called  Minorminer\footnote{\url{https://docs.ocean.dwavesys.com/projects/minorminer/en/latest/}} (developed by D-Wave) to embed the instance to $P_{100}$ Pegasus topology  (with $236,808$ qubits) and measure the number of qubits for the embedding. We set the time limit of the  embedding at one hour. 
%\begin{itemize}
%	\item \emph{Processing time.} We measure the time to compress the instances of algorithms. We exclude the time to load the network. 
%	\item \emph{Reduction ratio.} We compute the reduction ratio base on the number of physical qubits on D-Wave's Pegasus. More concretely, let $n_b, n_a$ be the number of physical qubits that are needed to solve an instance before and after compression, respectively.
%	The reduction ratio is computed as $\left(1 - \frac{n_a}{n_b}\right) \times 100\%$. 
%\end{itemize}

%\td{Adding how we generate Pegasus graph and embed  (add description, citation and one picture of pegasus graph) - correlation between the number of logical qubits \#edges vs. physical qubits}

\noindent \textbf{Environment.} We implemented our algorithms in C++ and obtained the implementations of others from the corresponding authors. We conducted all experiments on a CentOS machine Intel(R) Xeon(R) CPU E7-8894 v4 2.40GHz.

%\td{Objective:}
%Our experiments on real-world instances with a simple exact compression of size three (k = 3) is shown in the tables below. The results indicate a significant reduction in the number of logical spins. 
%
%\td{Experimental Setup}
%
%\td{Algorithms}
%We compare our methods with the following algorithms
%
%
%\td{QUBO instances}
%We benchmark the algorithms on the both synthetic QUBO instances and  real-world instances derived from quantum computing applications \cite{}.
%
%
%\td{Metrics}
%
%\td{Results on Synthetic instances} 
%
%\td{Real-world Applications}
%
%The relative difference on reduction ratios: Let $x,y$ be the reduction ration of two algorithm on the same network. We compute the relative difference as 
%$ \frac{x-y}{\max(x,y)}.$ If $x = y = 0$, the relative difference is $0$.

\subsection{Benchmark on synthetic instances}

%\subsubsection{Erdos-Renyi networks}
\begin{figure}[!ht]
%\vspace{-0.1in}
	\centering 	
	%	\begin{subfigure}{0.3\textwidth}
	%		\includegraphics[width=\linewidth]{figures/legend1.pdf}
	%	\end{subfigure}	
	
	\begin{subfigure}{0.24\textwidth}
		\includegraphics[width=\linewidth]{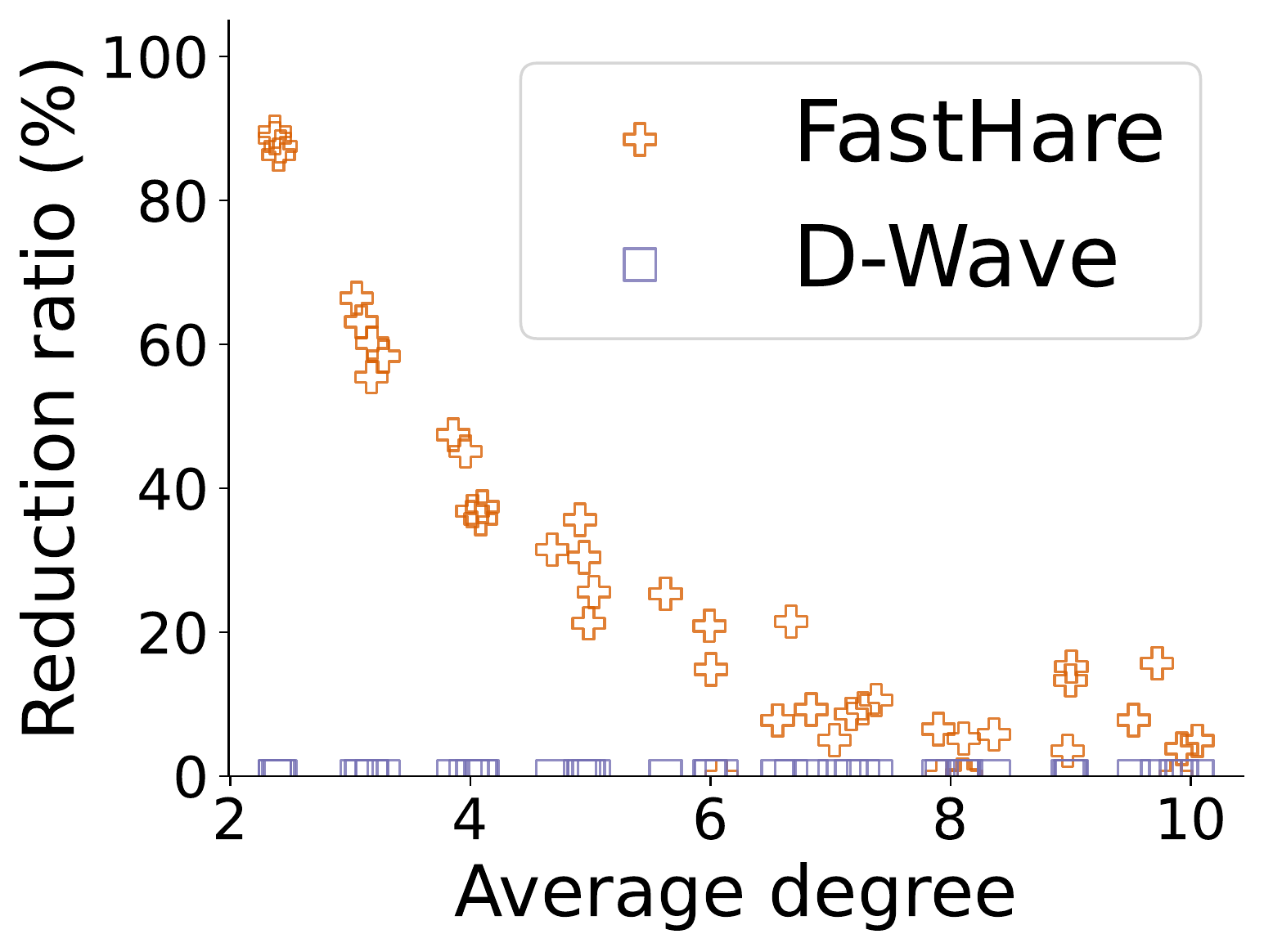}
		\caption{Erdos-Renyi networks} 
	\end{subfigure}	
	\begin{subfigure}{0.24\textwidth}
		\includegraphics[width=\linewidth]{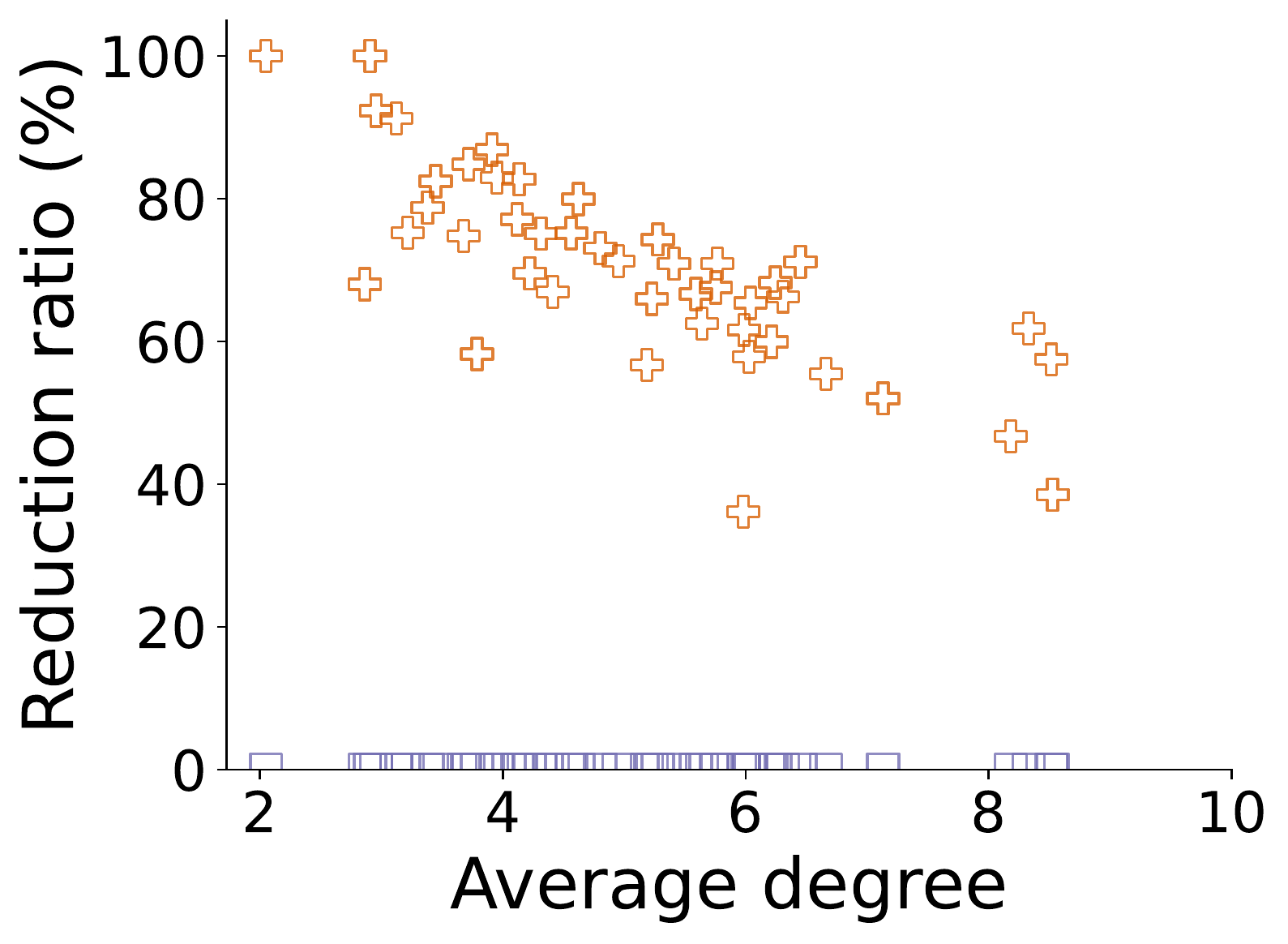}
		\caption{Scale-free networks} 
	\end{subfigure}	
	
	\caption{Reduction ratio on physical qubits (the higher is better). FastHare provides significant reduction on instances of different sizes with more reduction towards sparser instances. The implemented reduction in D-Wave's SDK offers no reduction for any instances.  
		\label{fig:reduction}} 
\end{figure}

%\begin{figure}[!ht]
%	\centering 	
%	\begin{subfigure}{0.3\textwidth}
%		\includegraphics[width=\linewidth]{figures/legend1.pdf}
%	\end{subfigure}	
%	
%	\begin{subfigure}{0.24\textwidth}
%		\includegraphics[width=\linewidth]{figures/reduction.pdf}
%		\caption{Reduction:  variables} 
%	\end{subfigure}	
%	\begin{subfigure}{0.24\textwidth}
%		\includegraphics[width=\linewidth]{figures/reduction-pegasus.pdf}
%		\caption{Reduction: physical qubits (D-Wave's Pegasus)} 
%	\end{subfigure}	
%	
%	\caption{Summary
%		\label{fig:summary1}} 
%\end{figure}
%\begin{figure}[!ht]
%	\centering 	
%	
%	\begin{subfigure}{0.24\textwidth}
%		\includegraphics[width=\linewidth]{figures/MQLib-reduction.pdf}
%		\caption{Reduction ration} 
%	\end{subfigure}	
%	\begin{subfigure}{0.24\textwidth}
%		\includegraphics[width=\linewidth]{figures/MQLib-time.pdf}
%		\caption{Processing time} 
%	\end{subfigure}	
%	
%	\caption{ Comparison of FastHare and Dwave roof duality.
%		\label{fig:compare}} 
%\end{figure}
%
%\begin{figure}[!ht]
%	\centering 	
%	\includegraphics[width=0.5\linewidth]{figures/MQLib-reduction-his.pdf}
%	\caption{Histogram
%		\label{fig:his}} 
%\end{figure}
\begin{table*}[h!]
\vspace{-0.05in}
	\begin{center}
		\begin{tabular}{p{1.7cm}|r|r|r|rr|rr|rr|rr} 
			{\multirow{3}{*}{Problem}} & \multirow{3}{*}{\#tests} & \multirow{3}{*}{\#nodes} & \multirow{3}{*}{Deg.}& 
			\multicolumn{2}{c|}{\multirow{2}{*}{Avg. processing time}} & 			\multicolumn{2}{c|}{\multirow{2}{*}{\#reducible instances}} & 	\multicolumn{4}{c}{Reduction ratio}\\
			\cline{9-12}
			&  &  & & 	& & &  & \multicolumn{2}{c|}{Logical qubits} & \multicolumn{2}{c}{Physical qubits}\\
			\cline{5-12}
			&  &  & & FastHare& D-Wave & FastHare  & D-Wave & FastHare  & D-Wave & FastHare & D-Wave
\\			\hline

			Gset \cite{gset}& 17 &  5k-20k & 2-12 &\textbf{ 0.0s} & 0.1s& \textbf{5} & 3 &\textbf{ 6\% ({\color{blue}19\%})} & 0\% ({\color{blue}2\%}) &  NA ({\color{blue}NA}) & NA ({\color{blue}NA})\\
			Beasley \cite{beasley1990or}& 60 & 	0k-3k &6-250 & \textbf{0.0s} & 0.1s& \textbf{20} & 3 & \textbf{8\% ({\color{blue}24\%})} & 0\% ({\color{blue}2\%}) & \textbf{10\% ({\color{blue}31\%})} & 1\% ({\color{blue}4\%})\\
			Culberson  \cite{culberson1995hiding} & 108 & 	 1k-5k &4-2,927 & \textbf{0.1s} & \textbf{0.1s}& \textbf{57} & 0 &\textbf{ 8\% ({\color{blue}15\%})} & 0\% ({\color{blue}0\%}) & \textbf{28\% ({\color{blue}31\%})} & 0\% ({\color{blue}0\%})\\
			Imgseg \cite{sousa2013estimation} & 100 & 	1k-28k & 2-5 & \textbf{0.1s} & 0.2s& \textbf{100} & 0 & \textbf{79\% ({\color{blue}79\%}) }& 0\% ({\color{blue}0\%}) &\textbf{ 92\% ({\color{blue}92\%})} & 0\% ({\color{blue}0\%})\\
			Others & 3,111 & 	0k-38k & 1-6,965 & \textbf{0.3s} & 0.5s& \textbf{1,302} & 296 & \textbf{21\% ({\color{blue}50\%}) }& 8\% ({\color{blue}82\%}) &\textbf{ 35\% ({\color{blue}62\%})} & 10\% ({\color{blue}84\%})\\
			\hline
			Overall & 3,396 & 	0k-38k & 1-6,965 &\textbf{ 0.3s} & 0.5s& \textbf{1,484} & 302 & \textbf{22\% ({\color{blue}51\%})} & 7\% ({\color{blue}81\%}) & \textbf{36\% ({\color{blue}62\%})} & 10\% ({\color{blue}84\%})\\
		\end{tabular}
	\end{center}
%	\vspace{-0.05in}
	\caption{Comparison on real world problems. Here, we can only embed $2,031$ instances with in an hour. The reduction ratio of physical qubits is reported based on those instances.}
	\label{table:app}
	%		$30$ max-cut instances in \cite{liers2004contributions,liers2004computing}.}
%\vspace{-0.05in}
\end{table*}
  
% The reduction ratio of QPBO is close to $0\%$ on all of the instances.
\begin{figure}[!ht]
\vspace{-0.2in}
	\centering 	
	\begin{subfigure}{0.24\textwidth}
		\includegraphics[width=\linewidth]{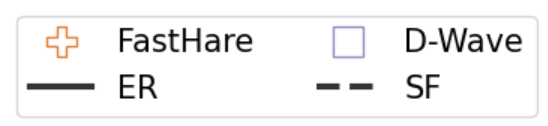}
	\end{subfigure}	
	
	\begin{subfigure}{0.24\textwidth}
		\includegraphics[width=\linewidth]{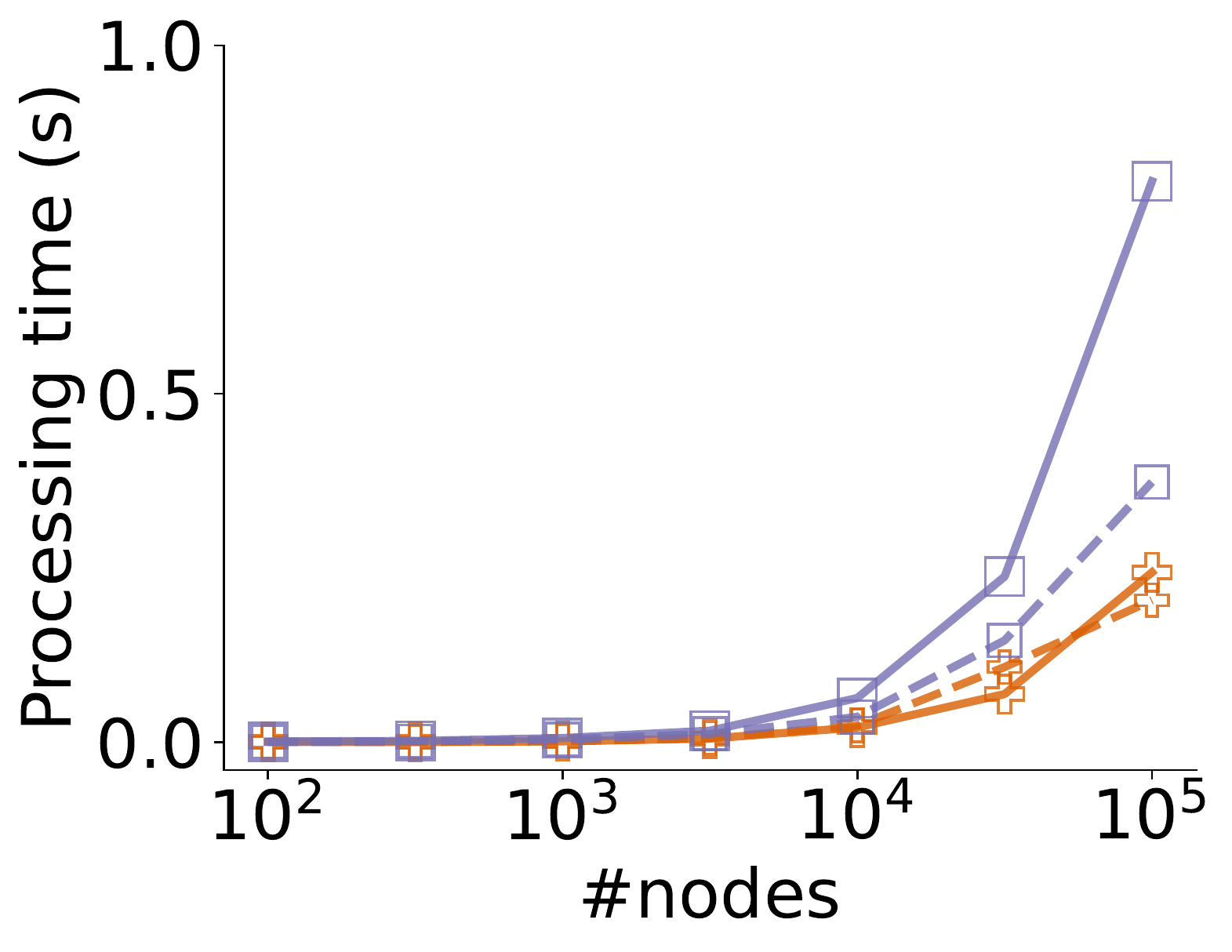}
		\caption{Varying network size} 
	\end{subfigure}	
	\begin{subfigure}{0.24\textwidth}
		\includegraphics[width=\linewidth]{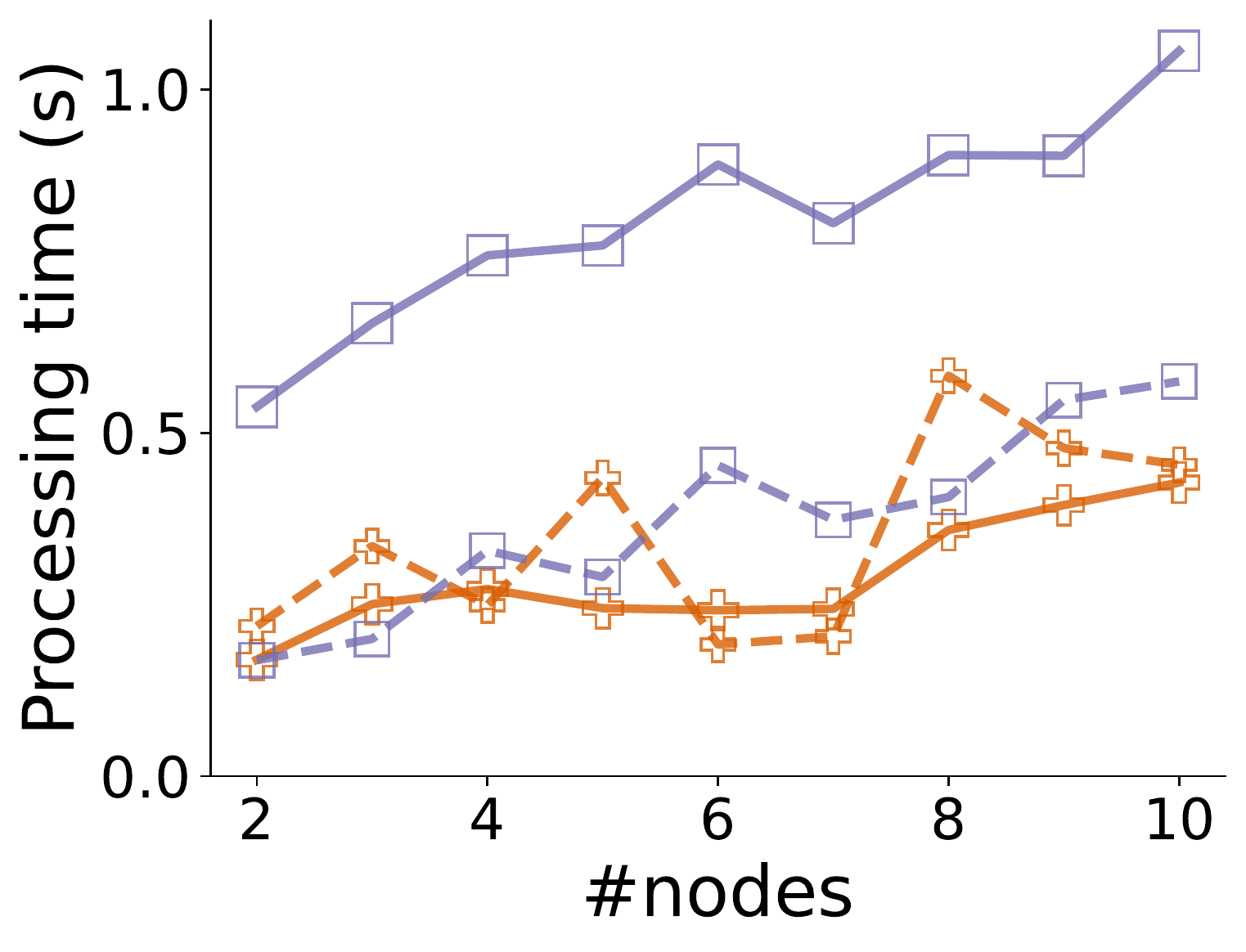}
		\caption{Varying network density} 
	\end{subfigure}

	\caption{Processing time in seconds on Erdos-Renyi (ER) networks, marked with solid lines, and scale-free (SF) networks, marked with dashed lines. 
		\label{fig:summary}} 
%	\vspace{-0.05in}	
\end{figure}

\noindent \textbf{Reduction ratio.} In Fig. \ref{fig:reduction}, we show the reduction ratio on physical qubits for FastHare and D-wave.  
 The roof duality implemented in D-Wave's SDK offers \emph{no reduction} for any instances. In contrast, FastHare can reduce all instances with the average reduction ratios  on Erdos-Renyi  and scale-free networks of $29\%$ and $67\%$, respectively. 

The reduction ratio gets lower quickly when the average degree increases. It suggests dense Ising Hamiltonians are generally harder to reduce. In addition, the instances with Erdos-Renyi topology are much harder to reduce comparing to the ones with scale-free topology. This suggests that random Hamiltonian with Erdos-Renyi topology of high degree contain less `redundant' information and, thus, can be used as hard benchmark instances for quantum solvers.

\noindent \textbf{Processing time.} 
Based on Fig. \ref{fig:summary}, the processing time of FastHare is several folds faster than D-Wave's. For example, on the largest Erdos-Renyi network with the number of nodes $n = 100,000$, the running time of FastHare and D-Wave  are $0.2$s and $0.8$s, respectively. Nevertheless, in terms of processing time, both roof duality implemented in D-Wave and FastHare are  highly efficient in preprocessing Hamiltonian before mapping to the QPU.  

\subsection{Benchmark on MQLib instances}
%\td{
%For each applications, provide 1) Summary of the problem with citation, 2) Formulation with citation 3) Description of the datasets (how are they generated, nodes, edges, etc.). + Characteristics of the QUBO instances (avg. degree, sparse or dense, weight distributions, etc.) 4) A table for the results
%}
Our experiments on MQLib \cite{DunningEtAl2018} is shown in Table \ref{table:app}. The results indicate a significant reduction by FastHare algorithm. FastHare outperforms D-Wave in the reduction ratio.
It can reduce  $1,484$ out of $3,396$ instances, i.e., about 5 times more than that of D-Wave's roof duality. The average reduction ratio in terms of logical and physical qubits among the reducible instances are $51\%$ and $62\%$, respectively. It suggest a significant qubit savings as a 62\% reduction mean we can solve instances that  require 2.5 times more qubits than the current limit on the state-of-the-art quantum annealers.
%in the number of logical spins. 
%For the physical qubits, we are able to embed $2,113$ instances. For those instances, we can reduce the size of $1217$ instances and the average reduction ratio is $62\%$. 

% \vspace{-0.08in}
\subsection{Reducibility prediction}
We investigate 70 metrics that are provided in the MQLib\footnote{\url{https://github.com/MQLib/MQLib/blob/master/data/metrics.csv}} to see which characteristics   affect the reducibility of the instances. 
We rank the metrics based on the Pearson correlation coefficient \cite{benesty2009pearson} with the reduction ratio.
Top $5$ characteristics with the highest correlation for FastHare and D-Wave are shown in Table \ref{table:reduction}. 

%We compute the Pearson correlation coefficient of the reduction ratios with each metric. 
The top two metrics (log norm ev2 and log norm ev1, respectively) are all calculated from
the weighted graph Laplacian matrix: the logarithm of the first and second largest eigenvalues normalized by
the average node degree  and the logarithm of the ratio
of the two largest eigenvalues (log ev ratio). This suggests that Hamiltonian with sparse cut are easier to reduce for FastHare. 

The implemented D-Wave's roof duality seems to work well on instances with constant clustering coefficient ({clust\_const}).  This behavior requires further investigation to determine the true reason behind why D-Wave's roof duality works very well on a few instances but cannot compress for the rest.

\begin{table}[h]
%\vspace{-0.1in}
	\begin{center}
		\begin{tabular}{lr|lr} 
			\multicolumn{2}{c|}{FastHare}	& \multicolumn{2}{c}{D-Wave}\\
			Metrics & Corr.  & Metrics & Corr. \\
			\hline
			log\_norm\_ev2	&0.73&	clust\_const&	0.42\\
			log\_norm\_ev1	&0.66	&clust\_log\_kurtosis	&-0.35\\
			mis	&0.59&	clust\_max&	-0.27\\
			log\_ev\_ratio&	-0.47&	weight\_mean&	0.25\\
			clust\_stdev	&0.46&	mis&	0.25
		\end{tabular}
	\end{center}
	\caption{Top 5 metrics with the highest correlation with reduction ratio.}
	\label{table:reduction}
	%		$30$ max-cut instances in \cite{liers2004contributions,liers2004computing}.}
	%\vspace{-0.1in}
\end{table}

We also use logistic regression \cite{wright1995logistic} to identify the metrics that have the most effect on the reducibility of the instances. 
Here, we remove $12$ time related metrics and normalize the remaining metrics such that the maximum absolute value of each metric equal one. For each algorithm, we set the label of an instance to one if the algorithm can reduce that instance. After running the logistic regression, we normalize the weights of the logistic regression such that the norm two of the weight vector equal one.
Table \ref{table:reg} shows  the top $5$ metrics with the highest absolute weights. 

%\vspace{-0.3cm}
\begin{table}[h]
	\begin{center}
		\begin{tabular}{lr|lr} 
			\multicolumn{2}{c|}{FastHare}	& \multicolumn{2}{c}{D-Wave}\\
			Metrics & Weight  & Metrics & Weight \\
			\hline
			chromatic&0.33&mis&0.49\\
			weight\_log\_kurtosis&0.29&weight\_max&-0.39\\
			mis&0.27&avg\_neighbor\_deg\_mean&-0.26\\
			avg\_neighbor\_deg\_mean&-0.24&core\_log\_kurtosis&0.25\\
			log\_ev\_ratio&-0.20&percent\_pos&-0.23\\
			
		\end{tabular}
	\end{center}
	\caption{Top 5 metrics that have the highest absolute weights in the logistic regression.}
	\label{table:reg}
	%		$30$ max-cut instances in \cite{liers2004contributions,liers2004computing}.}
	%\vspace{-0.1in}
\end{table}

\balance
\section{Conclusion}
\label{sec:conclusion}
We propose FastHare, an algorithm to reduce the size of Ising Hamiltonian, thus,  provide qubits saving for quantum annealing. The method is generic and can be applied for Ising Hamiltonian of different applications. We perform the first large-scale benchmarks to measure the reducibility in 3000+ instances from MQLib library and synthesized Hamiltonian, showing significant saving in applying Hamiltonian reduction. Importantly, the fast processing time of FastHare (averaging 0.3s) make it an inexpensive choice for preprocessing. FastHare also outperforms the roof duality reduction, implemented in D-Wave's Ocean SDK,    both in time and quality by several folds. In future,  FastHare can be integrated with minor-embedding methods to balance between number of physical qubits, chain lengths, and range of the coupling strengths to further improve the performance of quantum solvers.

\bibliographystyle{IEEEtran}
\balance

%\bibliography{compress,thang}

\begin{thebibliography}{10}
\providecommand{\url}[1]{#1}
\csname url@samestyle\endcsname
\providecommand{\newblock}{\relax}
\providecommand{\bibinfo}[2]{#2}
\providecommand{\BIBentrySTDinterwordspacing}{\spaceskip=0pt\relax}
\providecommand{\BIBentryALTinterwordstretchfactor}{4}
\providecommand{\BIBentryALTinterwordspacing}{\spaceskip=\fontdimen2\font plus
\BIBentryALTinterwordstretchfactor\fontdimen3\font minus
  \fontdimen4\font\relax}
\providecommand{\BIBforeignlanguage}[2]{{%
\expandafter\ifx\csname l@#1\endcsname\relax
\typeout{** WARNING: IEEEtran.bst: No hyphenation pattern has been}%
\typeout{** loaded for the language `#1'. Using the pattern for}%
\typeout{** the default language instead.}%
\else
\language=\csname l@#1\endcsname
\fi
#2}}
\providecommand{\BIBdecl}{\relax}
\BIBdecl

\bibitem{arute2019quantum}
F.~Arute, K.~Arya, R.~Babbush, D.~Bacon, J.~C. Bardin, R.~Barends, R.~Biswas,
  S.~Boixo, F.~G. Brandao, D.~A. Buell \emph{et~al.}, ``Quantum supremacy using
  a programmable superconducting processor,'' \emph{Nature}, vol. 574, no.
  7779, pp. 505--510, 2019.

\bibitem{honjo2021100}
T.~Honjo, T.~Sonobe, K.~Inaba, T.~Inagaki, T.~Ikuta, Y.~Yamada, T.~Kazama,
  K.~Enbutsu, T.~Umeki, R.~Kasahara \emph{et~al.}, ``100,000-spin coherent
  ising machine,'' \emph{Science advances}, vol.~7, no.~40, p. eabh0952, 2021.

\bibitem{bharti2022noisy}
K.~Bharti, A.~Cervera-Lierta, T.~H. Kyaw, T.~Haug, S.~Alperin-Lea, A.~Anand,
  M.~Degroote, H.~Heimonen, J.~S. Kottmann, T.~Menke \emph{et~al.}, ``Noisy
  intermediate-scale quantum algorithms,'' \emph{Reviews of Modern Physics},
  vol.~94, no.~1, p. 015004, 2022.

\bibitem{mills2021two}
A.~Mills, C.~Guinn, M.~Gullans, A.~Sigillito, M.~Feldman, E.~Nielsen, and
  J.~Petta, ``Two-qubit silicon quantum processor with operation fidelity
  exceeding 99\%,'' \emph{arXiv preprint arXiv:2111.11937}, 2021.

\bibitem{wang2022light}
X.~Wang, C.~Xiao, H.~Park, J.~Zhu, C.~Wang, T.~Taniguchi, K.~Watanabe, J.~Yan,
  D.~Xiao, D.~R. Gamelin \emph{et~al.}, ``Light-induced ferromagnetism in
  moir{\'e} superlattices,'' \emph{Nature}, vol. 604, no. 7906, pp. 468--473,
  2022.

\bibitem{kadowaki1998quantum}
T.~Kadowaki and H.~Nishimori, ``Quantum annealing in the transverse ising
  model,'' \emph{Physical Review E}, vol.~58, no.~5, p. 5355, 1998.

\bibitem{zhou2022noise}
Y.~Zhou and P.~Zhang, ``Noise-resilient quantum machine learning for stability
  assessment of power systems,'' \emph{IEEE Transactions on Power Systems},
  2022.

\bibitem{fox2021mrna}
D.~M. Fox, K.~M. Branson, and R.~C. Walker, ``mrna codon optimization with
  quantum computers,'' \emph{PloS one}, vol.~16, no.~10, p. e0259101, 2021.

\bibitem{mulligan2020designing}
V.~K. Mulligan, H.~Melo, H.~I. Merritt, S.~Slocum, B.~D. Weitzner, A.~M.
  Watkins, P.~D. Renfrew, C.~Pelissier, P.~S. Arora, and R.~Bonneau,
  ``Designing peptides on a quantum computer,'' \emph{BioRxiv}, p. 752485,
  2020.

\bibitem{fox2022rna}
D.~M. Fox, C.~M. MacDermaid, A.~M. Schreij, M.~Zwierzyna, and R.~C. Walker,
  ``Rna folding using quantum computers,'' \emph{PLOS Computational Biology},
  vol.~18, no.~4, p. e1010032, 2022.

\bibitem{mugel2021hybrid}
S.~Mugel, M.~Abad, M.~Bermejo, J.~S{\'a}nchez, E.~Lizaso, and R.~Or{\'u}s,
  ``Hybrid quantum investment optimization with minimal holding period,''
  \emph{Scientific Reports}, vol.~11, no.~1, pp. 1--6, 2021.

\bibitem{grozea2021optimising}
C.~Grozea, R.~Hans, M.~Koch, C.~Riehn, and A.~Wolf, ``Optimising rolling stock
  planning including maintenance with constraint programming and quantum
  annealing,'' \emph{arXiv preprint arXiv:2109.07212}, 2021.

\bibitem{yarkoni2021multi}
S.~Yarkoni, A.~Alekseyenko, M.~Streif, D.~Von~Dollen, F.~Neukart, and
  T.~B{\"a}ck, ``Multi-car paint shop optimization with quantum annealing,'' in
  \emph{2021 IEEE International Conference on Quantum Computing and Engineering
  (QCE)}.\hskip 1em plus 0.5em minus 0.4em\relax IEEE, 2021, pp. 35--41.

\bibitem{mott2017solving}
A.~Mott, J.~Job, J.-R. Vlimant, D.~Lidar, and M.~Spiropulu, ``Solving a higgs
  optimization problem with quantum annealing for machine learning,''
  \emph{Nature}, vol. 550, no. 7676, pp. 375--379, 2017.

\bibitem{neukart2017traffic}
F.~Neukart, G.~Compostella, C.~Seidel, D.~Von~Dollen, S.~Yarkoni, and
  B.~Parney, ``Traffic flow optimization using a quantum annealer,''
  \emph{Frontiers in ICT}, vol.~4, p.~29, 2017.

\bibitem{kim2019leveraging}
M.~Kim, D.~Venturelli, and K.~Jamieson, ``Leveraging quantum annealing for
  large mimo processing in centralized radio access networks,'' in
  \emph{Proceedings of the ACM Special Interest Group on Data Communication},
  2019, pp. 241--255.

\bibitem{dwaveAdvantage}
``D-wave hybrid solver service: An overview,''
  \url{https://www.dwavesys.com/solutions-and-products/systems/}, 2020.

\bibitem{choi2008minor}
V.~Choi, ``Minor-embedding in adiabatic quantum computation: I. the parameter
  setting problem,'' \emph{Quantum Information Processing}, vol.~7, no.~5, pp.
  193--209, 2008.

\bibitem{tabi2021evaluation}
Z.~I. Tabi, {\'A}.~Marosits, Z.~Kallus, P.~Vaderna, I.~G{\'o}dor, and
  Z.~Zimbor{\'a}s, ``Evaluation of quantum annealer performance via the massive
  mimo problem,'' \emph{IEEE Access}, vol.~9, pp. 131\,658--131\,671, 2021.

\bibitem{hammer1984roof}
P.~L. Hammer, P.~Hansen, and B.~Simeone, ``Roof duality, complementation and
  persistency in quadratic 0--1 optimization,'' \emph{Mathematical
  programming}, vol.~28, no.~2, pp. 121--155, 1984.

\bibitem{boros2006preprocessing}
E.~Boros, P.~L. Hammer, and G.~Tavares, ``Preprocessing of unconstrained
  quadratic binary optimization,'' 2006.

\bibitem{Rother2007OptimizingBM}
C.~Rother, V.~Kolmogorov, V.~S. Lempitsky, and M.~Szummer, ``Optimizing binary
  mrfs via extended roof duality,'' \emph{2007 IEEE Conference on Computer
  Vision and Pattern Recognition}, pp. 1--8, 2007.

\bibitem{lange2019combinatorial}
J.-H. Lange, B.~Andres, and P.~Swoboda, ``Combinatorial persistency criteria
  for multicut and max-cut,'' in \emph{Proceedings of the IEEE/CVF Conference
  on Computer Vision and Pattern Recognition}, 2019, pp. 6093--6102.

\bibitem{boros2002pseudo}
E.~Boros and P.~L. Hammer, ``Pseudo-boolean optimization,'' \emph{Discrete
  applied mathematics}, vol. 123, no. 1-3, pp. 155--225, 2002.

\bibitem{lucas2014ising}
A.~Lucas, ``Ising formulations of many np problems,'' \emph{Frontiers in
  physics}, p.~5, 2014.

\bibitem{finnila1994quantum}
A.~B. Finnila, M.~Gomez, C.~Sebenik, C.~Stenson, and J.~D. Doll, ``Quantum
  annealing: A new method for minimizing multidimensional functions,''
  \emph{Chemical physics letters}, vol. 219, no. 5-6, pp. 343--348, 1994.

\bibitem{farhi2000quantum}
E.~Farhi, J.~Goldstone, S.~Gutmann, and M.~Sipser, ``Quantum computation by
  adiabatic evolution,'' \emph{arXiv preprint quant-ph/0001106}, 2000.

\bibitem{albash2018adiabatic}
T.~Albash and D.~A. Lidar, ``Adiabatic quantum computation,'' \emph{Reviews of
  Modern Physics}, vol.~90, no.~1, p. 015002, 2018.

\bibitem{choi2011minor}
V.~Choi, ``Minor-embedding in adiabatic quantum computation: Ii.
  minor-universal graph design,'' \emph{Quantum Information Processing},
  vol.~10, no.~3, pp. 343--353, 2011.

\bibitem{10.5555/574848}
M.~R. Garey and D.~S. Johnson, \emph{Computers and Intractability; A Guide to
  the Theory of NP-Completeness}.\hskip 1em plus 0.5em minus 0.4em\relax USA:
  W. H. Freeman \& Co., 1990.

\bibitem{panchenko2013sherrington}
D.~Panchenko, \emph{The sherrington-kirkpatrick model}.\hskip 1em plus 0.5em
  minus 0.4em\relax Springer Science \& Business Media, 2013.

\bibitem{DunningEtAl2018}
I.~Dunning, S.~Gupta, and J.~Silberholz, ``What works best when? a systematic
  evaluation of heuristics for max-cut and {QUBO},'' \emph{{INFORMS} Journal on
  Computing}, vol.~30, no.~3, 2018.

\bibitem{schult2008exploring}
D.~A. Schult, ``Exploring network structure, dynamics, and function using
  networkx,'' in \emph{In Proceedings of the 7th Python in Science Conference
  (SciPy}.\hskip 1em plus 0.5em minus 0.4em\relax Citeseer, 2008.

\bibitem{gset}
\BIBentryALTinterwordspacing
``Gset.'' [Online]. Available: \url{http://web.stanford.edu/~yyye/yyye/Gset/}
\BIBentrySTDinterwordspacing

\bibitem{beasley1990or}
J.~E. Beasley, ``Or-library: distributing test problems by electronic mail,''
  \emph{Journal of the operational research society}, vol.~41, no.~11, pp.
  1069--1072, 1990.

\bibitem{culberson1995hiding}
J.~Culberson, A.~Beacham, and D.~Papp, ``Hiding our colors,'' in \emph{CP’95
  Workshop on Studying and Solving Really Hard Problems}.\hskip 1em plus 0.5em
  minus 0.4em\relax Citeseer, 1995, pp. 31--42.

\bibitem{sousa2013estimation}
S.~d. Sousa, Y.~Haxhimusa, and W.~G. Kropatsch, ``Estimation of distribution
  algorithm for the max-cut problem,'' in \emph{International Workshop on
  Graph-Based Representations in Pattern Recognition}.\hskip 1em plus 0.5em
  minus 0.4em\relax Springer, 2013, pp. 244--253.

\bibitem{boothby2020next}
K.~Boothby, P.~Bunyk, J.~Raymond, and A.~Roy, ``Next-generation topology of
  d-wave quantum processors,'' \emph{arXiv preprint arXiv:2003.00133}, 2020.

\bibitem{benesty2009pearson}
J.~Benesty, J.~Chen, Y.~Huang, and I.~Cohen, ``Pearson correlation
  coefficient,'' in \emph{Noise reduction in speech processing}.\hskip 1em plus
  0.5em minus 0.4em\relax Springer, 2009, pp. 1--4.

\bibitem{wright1995logistic}
R.~E. Wright, ``Logistic regression.'' 1995.

\end{thebibliography}
%\appendix
%\input{appendix}
\end{document}